\documentclass[8pt]{article}

\usepackage{graphicx}
\usepackage{psfrag} 
\usepackage{amsthm}
\usepackage{mathrsfs} 
\usepackage{xcolor}
\usepackage{mathtools}
\mathtoolsset{showonlyrefs=true}   
\usepackage{authblk}  
 \usepackage[backend=biber,citestyle = alphabetic, bibstyle = plain,  sortlocale=en_US, sorting=nyt, hyperref=auto, style=numeric,defernumbers=true, isbn=false, doi=false,url=false, bibencoding = utf8]{biblatex}
 \bibliography{bibliografia.bib}     
\def\g{{\mathfrak g}}
\def\ra{\rangle}
\def\la{\langle} 
\def\e{{\bm{e}}}

\def\M{\mathcal O}
\def\hep{\hat{\epsilon}}
\def\T{{\mathcal T}}
\def\co{\mbox{Co}}
\def\coh{\mbox{Coh}}
\def\si{\mbox{Si}} 
\def\sih{\mbox{Sih}}

\def\coexp{\cos(\sqrt{Q} \ell_1)}
\def\cohexp{\cosh(\sqrt{Q} \ell_2)}
\def\siexp{\sin(\sqrt{Q} \ell_1)}
\def\sihexp{\sinh(\sqrt{Q} \ell_2)}

\def\a{{\mathfrak a}}
\def\b{{\mathfrak b}}
\def\w{{\mathfrak w}}

\def\al{\alpha}
\def\be{\beta}
\def\ga{\gamma}
\def\del{\delta}

\def\alb{\overline{\alpha}}
\def\beb{\overline{\beta}}
\def\gab{\overline{\gamma}}
\def\delb{\overline{\delta}}

\def\s{\mathfrak z}
\def\aa{\sigma}
\def\bb{\tau}
\def\cc{\sigma}
\def\arg{\theta}
\def\ca{A}
  \def\cb{B}
\def\z{z}
\def\zb{\overline{z}}
\def\A{\mathbb{A}}
\def\Q{\mathbb{Q}}
\def\Mo{\mbox{Moeb}}

\newcommand{\bm}[1]{\mbox{\boldmath $#1$}}

\newcommand{\lr}[1]{\left( #1 \right)}
\newcommand{\lrbrace}[1]{\left\lbrace #1 \right\rbrace}
\newcommand{\lrprod}[1]{\left< #1 \right>}
\newcommand{\lrbrkt}[1]{\left[ #1 \right]}

\newcommand{\skwend}[1]{{\mathrm{SkewEnd}\lr{#1}}}

\newcommand{\Tr}{\mathrm{Tr}}
\newcommand{\spn}[1]{\mathrm{span}\left\lbrace #1 \right\rbrace}

\def\Fmuind{\stackrel{\mu}{F}}
\def\Fmu{\stackrel{\mu}{\bm{F}}}

 \newcounter{mnotecount}[section]
 
 \renewcommand{\themnotecount}{\thesection.\arabic{mnotecount}}

 \newcommand{\mnote}[1]%{}
 {\protect{\stepcounter{mnotecount}}$^{\mbox{\footnotesize
 $%\!\!\!\!\!\!\,
 \bullet$\themnotecount}}$ \marginpar{%\color{red}%
 \raggedright\tiny\em
 $\!\!\!\!\!\!\,\bullet$\themnotecount: #1} }

 \title{Skew-symmetric endomorphisms in $\mathbb{M}^{1,3}$: A unified canonical form with applications to conformal geometry}

\author{
  Marc Mars and Carlos Pe\'on-Nieto \\
  Instituto de F\'{\i}sica  Fundamental y Matem\'aticas, Universidad de Salamanca \\
Plaza de la Merced s/n 37008, Salamanca, Spain}

\date{}

\usepackage{amsfonts}
\usepackage{amsmath}
\usepackage{amssymb}

\newtheorem{lemma}{Lemma}
\newtheorem{definition}{Definition} 
\newtheorem{proposition}{Proposition}
\newtheorem{corollary}{Corollary}
\newtheorem{theorem}{Theorem}
\newtheorem{remark}{Remark}

\def\Fmuind{\stackrel{\mu}{F}}
\def\Fmu{\stackrel{\mu}{\bm{F}}}

\begin{document}

\maketitle

\begin{abstract}
  We derive a canonical form for skew-symmetric endomorphisms $F$ in Lorentzian vector spaces of dimension three and four which covers all non-trivial cases at once. We analyze its invariance group, as well as the connection of this canonical form with  duality rotations of two-forms.
  After reviewing the relation between these endomorphisms and the algebra of
  conformal Killing vectors of $\mathbb{S}^2$, $\mathrm{CKill}\lr{\mathbb{S}^2}$, we are able to also give a canonical form for an arbitrary element $ \xi \in \mathrm{CKill}\lr{\mathbb{S}^2}$ along with its invariance group. The construction allows us to obtain explicitly the change of basis that transforms any given $F$ into its canonical form. For any non-trivial $ \xi$ we construct, via its canonical form, adapted coordinates that allow us to study its properties in depth.   Two applications are worked out: we determine explicitly for which metrics, among a natural class of spaces of constant curvature, a given $\xi$ is a Killing vector and solve all local TT (traceless and transverse) tensors that satisfy the Killing Initial Data equation for $\xi$.  In addition to their own interest, the present results will be a basic ingredient for a subsequent generalization to arbitrary dimensions.
    \end{abstract}

\section{Introduction}

Finding a canonical form for the elements of a certain set is often  an interesting problem to solve, since it is a powerful tool for  both computations and mathematical analysis. By canonical form (sometimes also called normal form) of the elements $x$ of a set $X$ one usually understands a specific form, depending on a number of parameters, that every element $x$ can be carried to. The value of such parameters is obviously determined by $x$. The most common examples are canonical forms of matrices, such as the echelon form or the Jordan form. However, the same concept arises in other sets, such as smooth fields on a manifold or even systems of differential equations (e.g. canonical coordinates for Hamiltonian systems). A canonical form must be somehow useful either to simplify the calculations or to make explicit some information we may want to exploit. Taking an element to its canonical form requires showing the existence of  (and ideally also finding explicitly) a transformation, namely, a change of basis, coordinates, etc. that brings the element into its canonical form, and 
which need not to be unique.

When dealing with Lie algebras $\mathfrak{g}$, one may attempt to find a canonical form for every element  $F \in \mathfrak{g}$  that captures all the information of its orbit under the (e.g. adjoint) action of the Lie group $G$. For example, this is the case of the already mentioned Jordan canonical form, regarded as the matrix form (up to permutations of the blocks) that encodes all the information of the $\mathfrak{gl}(n,\mathbb{C})$ orbits under the adjoint action of the group $GL(n,\mathbb{C})$.  Identifying these orbits, and the more general problem of the orbits generated by an algebraic group action on a set,  is an active field of research in different fields of mathematics and it is already well-understood for the case of classical Lie groups. We refer the reader to  \cite{complexorbs} and references therein for an extensive review of this problem and other references such as \cite{abellanas75}, \cite{burgoyne77}, \cite{djokovic83}, \cite{goodmanwallach}, \cite{knapp}.

From the point of view of physics, it is of particular importance the study of the pseudo-orthogonal group $O(1,n+1)$ because of its role in the theory of relativity and other physical theories. First, it is the group of isotropies
%pointwise isometries \footnote{We mean by pointwise that the actual group of isometries is $O(1,n+1)$ semidirect product with the group of translations $\mathbb{R}^{n+1}$, i.e. the Poincar\'e group}
in the special theory of relativity and isotropy group of the Lorentz-Maxwell electrodynamics. For the latter, the elements of the Lie algebra $\mathfrak{o}(1,n+1)$, represented here as skew-symmetric endomorphisms of Minkowski $\mathbb{M}^{1,n+1}$ (or equivalently the two-forms of the same space), also represent the electromagnetic field (e.g. \cite{LichnerowiczTheoRelGravEM}).  
Besides, and this is of great importance in our approach, in general relativity the pseudo-orthogonal group is related to the group of conformal transformations of certain spaces  \cite{PenroseRindVol1},\cite{IntroCFTschBook}.
Also, techniques in conformal geometry allow to recast the
  Einstein field equations (in fact, an equivalent set thereof)  as a Cauchy or characteristic problem in a hypersurface $\mathscr{I}$ (\cite{Friedrich2002}, \cite{Friedrich2014} and references therein) representing ``infinity'' in a physically precise sense. We are specially interested in the case of positive cosmological constant, where this Cauchy problem is always well-posed and $\mathscr{I}$ happens to be Riemannian. The initial data consist of a
 metric $\gamma$  in $\mathscr{I}$ and a symmetric ``TT'' tensor $D$ of $\mathscr{I}$, i.e. traceless and transverse (zero divergence). If the solution spacetime is to have a Killing vector, then the so called Killing Initial Data (KID) equations must be satisfied \cite{KIDPaetz}, and this involves  a conformal Killing vector (CKV) of  $\gamma$. Moreover, only the conformal class of the data matters and of particular importance is the case of $\gamma$ conformal to the standard metric of the sphere, in particular because of its relation with black hole spacetimes such as Kerr-de Sitter \cite{Kdslike}. We will expand on this later in this introduction.

In the physics literature a ``canonical'' form for the $\mathfrak{o}(1,n+1)$ elements is often employed mostly in four dimensions \cite{syngeSR} but also in arbitrary dimensions \cite{Kdslike}, \cite{ida20}. This form requires identifying the causal character of the eigenvectors of a given element $F \in \mathfrak{o}(1,n+1)$ and gives rise to two different types of canonical forms, one and only one admitted by each given $F$. Something similar is done in  more generality in  \cite{djokovic83} where, from a powerful classification result, a list of canonical forms for a wide sample of Lie algebras is given, but the pseudo-orthogonal case still requires two different forms. All these forms contain sufficient information to identify the orbit generated by the adjoint action of the group acting on the given element. 
However, it is surprising that, to the best of the authors' knowledge, there are no previous attempts to find
a unified canonical form to which any single element of the algebra  $\mathfrak{o}(1,n+1)$ can be carried to.  In the present paper, we address and solve the problem of finding a unified canonical form for skew-symmetric endomorphisms in three ($n=1$) and four ($n=2$) dimensions. 

As mentioned above, one aspect of the relevance of pseudo-orthogonal groups (or any signature) lies in their relation with the conformal group of a related space. For $O(1,n+1)$ this is the conformal group of the sphere $\mathbb{S}^{n}$, that we denote $\mathrm{Conf}\lr{\mathbb{S}^{n}}$.
More specifically, the orthochronous subgroup (i.e. the one preserving time orientation) $O^+(1,n+1)$ is isomorphic to $\mathrm{Conf}\lr{\mathbb{S}^{n}}$  \cite{Kdslike}, \cite{PenroseRindVol1}, and so it is the lie algebra $\mathfrak{o}(1,n+1)$ to the conformal Killing vector (CKV) fields $\mathrm{CKill}\lr{\mathbb{S}^{n}}$. Thus, finding a canonical form for the elements of $\mathfrak{o}(1,n+1)$, in turn implies a canonical form for the elements of $\mathrm{CKill}\lr{\mathbb{S}^{n}}$. Amongst other applications, it is particularly useful to employ the canonical form to find adapted coordinates to an arbitrary $\xi \in \mathrm{CKill}\lr{\mathbb{S}^{n}}$. In these coordinates, the KID equations are straightforward to solve with generality, which is a first step in order to obtain all TT tensors that generate spacetimes with at least one symmetry. This is a possible route to obtain a new characterization result for Kerr-de Sitter, specially relevant for the physical $n=3$ case. Here we study in detail  the $n = 2$ case, where in addition  we prove that there always exist an element $\xi^\perp \in \mathrm{CKill}\lr{\mathbb{S}^{n}}$, which is everywhere orthogonal to $\xi$, with the same norm and  such that $[\xi, \xi^\perp] = 0$ (c.f. Lemma \ref{orto} below), so it is convenient to 
adapt coordinates simultaneously 
to $\xi$, $\xi^\perp$. With these coordinates at hand, we obtain all TT tensors satisfying the KID equation in a very simple and elegant form (c.f. Section \ref{secapps}).

Some of the results here are generalizable to arbitrary dimension. However, we believe that the low dimensional case deserves a separate analysis for several reasons. First, the most relevant physical dimension for a spacetime is four, so studying this case in detail is particularly important and intrinsically interesting.
Second, the results presented here are an essential building block for the generalization to arbitrary dimensions. For example, a canonical form for any dimension  will follow by combining the results in this paper and well-known classification theorems of pseudo-orthogonal algebras. In addition, dealing with low dimensions allows us to analyze some of the questions in more depth and get additional insights into the problem. 
This perspective also provides  clues about  the possible solutions to the problem in more dimensions. Finally, although simpler than in higher dimensions, even the low dimensional case is far from trivial, so it makes sense the present this case in a separate work.
%, so for the sake of conciseness, we leave the task of generalizing our results here for a future paper.

This paper is intended to be self-contained  and only requires elementary knowledge of algebra and differential geometry. Our intention is to make our results accessible for readers with different backgrounds. The paper is organized as follows. Sections \ref{seccanon}, \ref{secinvg} and \ref{secgeninvg} are devoted to the obtention and analysis of a canonical form for any given (non-zero) element  $F \in \mathfrak{o}(1,3)$. In section \ref{seccanon} we obtain our canonical form in four dimensions, i.e. for  $\mathfrak{o}(1,3)$ and show its universal validity for every non-trivial $F$.
The change of basis that yields to the canonical form is not unique. This implies the existence of an invariance group, that we derive in  section \ref{secinvg}. In section \ref{secgeninvg} we analyze the generators of the invariance group and obtain a decomposition of the element $F$ in terms of these. We also make a connection between this decomposition and the standard duality rotations
  for two-forms.
In all these sections, the three-dimensional case is obtained and discussed as a corollary of the four-dimensional one.

The following Sections \ref{secGCKVs}, \ref{GCKV_and_F}, \ref{seccanonGCKV}, \ref{secadaptedcoords} are devoted to the study of so-called global CKVs (GCKV) defined on Euclidean space $\mathbb{E}^2$, and which are directly related to CKV on the sphere $\mathbb{S}^2$. Section \ref{secGCKVs} defines such vectors and Section \ref{GCKV_and_F} describes a known  relation
between them and the Lie algebra $\mathfrak{o}(1,3)$. In Section \ref{seccanonGCKV} we apply all the results for the $\mathfrak{o}(1,3)$ algebra to the CKVs of the sphere, namely, the obtention of a canonical form and its invariance group. As a useful consequence of the two viewpoints, we are able (Corollary \ref{canonbasescoord}) to obtain in a fully explicit form the change of basis that transforms any given $F$ into its canonical form.
Finally, Section \ref{secadaptedcoords} gives a set of coordinates adapted to an arbitrary $\xi$ and its orthogonal $\xi^\perp$. The results concerning the canonical form of GCKV and the adapted coordinates are summarized in Theorem \ref{Main}.
%(from Lemma \ref{orto} mentioned above). 
 Our last Section \ref{secapps} gives two interesting applications for the previous results. First, given a GCKV $\xi$, Theorem \ref{Killing} gives a list of all metrics, conformal to the metric of a 2-sphere, for which $\xi$ is a Killing vector. Second, Theorem \ref{theoTT} gives an elegant solution of the TT tensors satisfying the KID equations in $\mathbb{S}^2$.

\section{Canonical form of skew-symmetric endomorphisms in
  $\mathbb{M}^{1,3}$}\label{seccanon}

\label{CanonSkew}

In this section we consider Lorentzian four-vector spaces $(V,g)$, i.e a four dimensional vector space $V$ endowed with a pseudo-Riemannian metric $g$ of signature $\{-,+,+,+\}$. The inner product with $g$ is denoted by $\lrprod{\cdot,\cdot}$.  We will often identify Lorentzian vector spaces of dimension $n$ with Minkowski $\mathbb{M}^{1,n-1}$. Null vectors are vectors with vanishing norm (in particular, the zero vector is null in our conventions). An endomorphism
$F: V \longrightarrow V$ is skew-symmetric when it satisfies
\begin{align}\label{defskew}
  \la e, F(e') \ra = - \la F(e), e' \ra, \quad \quad \forall e,e' \in V.
\end{align}
This subset of $\mathrm{End}\lr{V}$ is denoted by $\skwend{V}$.
% Another property that will be used often is that 
% a skew-symmetric endomorphism $F \in \mbox{End}(V)$
% that restricts to a vector subspace $U \subset V$ (i.e. $F(U) \subset U$)
% also restricts to the orthogonal space $U^{\perp}$. Indeed, for
% $\x \in U, \y \in U^{\perp}$ 
% \begin{align*}
% \la F(\y),\x \ra = - \la \y, F(\x) \ra =0
% \end{align*}
% since $F(\x) \in U$. But since this is true for all $\x \in U$, this equation states precisely that $F(\y) \in U^{\perp}$.
We take, by definition, that eigenvectors of an endomorphism are always non-zero.
$\ker F$
and $\mathrm{Im\,} F$ denote, respectively, the kernel and image of
$F \in \mathrm{End}\lr{V}$.

We now briefly discuss a few basic properties of skew-symmetric endomorphisms that we will be referring to.
First, it is immediate from %the definition of skew-symmetry
\eqref{defskew} that every vector $e \in V$ is perpendicular to its image,
%under $F$,
i.e. $\lrprod{F(e),e} = 0$. Second, consider a, possibly complex, eigenvalue $\lambda \in \mathbb{C}$ and its eigenvector $w \in V_\mathbb{C}$ (the complexification of $V$). By the previous property, 
$w$ must be null if $\lambda \neq 0$, because $\lrprod{F(e),e} = \lambda \lrprod{e,e} = 0$. Eigenvectors with zero eigenvalue may be both null and non-null. Since $F$ is real, the complex conjugate $\lambda^\star \in \mathbb{C}$
  is an eigenvalue with eigenvector $w^\star \in V_\mathbb{C}$, so
%
%  Non-null (as well as null) eigenvectors are also possible as long as $\lambda =0$. If $\lambda \neq 0$, the complex conjugates
%  
%  and $
%  $ are also eigenvector and eigenvalue so 
 \begin{equation}
               \lrprod{F(w),w^\star} = \lambda \lrprod{w,w^\star} = - \lambda^\star \lrprod{w,w^\star}. 
              \end{equation}
              Thus, either $\lambda$ is purely imaginary (including zero) or, if not,  $w,w^\star$ are a pair of null vectors orthogonal to each other. Suppose the later and denote $w = u + i v$ for $u,v \in V$. Then the nullity condition implies $\lrprod{u,v} = 0$ and $\lrprod{u,u} = \lrprod{v,v}$ and orthogonality to $w^\star$ implies $\lrprod{u,u} = - \lrprod{v,v}$. Hence $u,v$ are null
              and proportional, i.e. $u = a v$ for some $a \in \mathbb{R}$, in consequence $w = (a + i) v$. Therefore, $v \in V$ is a real null eigenvector and its corresponding eigenvalue $\lambda$  must be real.  Summarizing, $F$ has only real or purely imaginary eigenvalues and their corresponding eigenvectors must be null for non-zero eigenvalues.

It will be useful to work with two-dimensional subspaces which are invariant under the action of $F$, which we will call ``eigenplanes''. Let $\spn{e,e'} = \Pi$ be a spacelike eigenplane for a pair of spacelike, orthogonal, unit vectors $e,e'$. Then by $F$-invariance
 \begin{equation*}
  F(e) = a_1 e + a_2 e',\quad \quad F(e') = b_1 e+ b_2 e',\quad \quad a_1,a_2,b_1,b_2 \in \mathbb{R}, 
 \end{equation*}
\and by skew-symmetry $a_1 =\lrprod{F(e),e} = 0,~b_2 =\lrprod{F(e'),e'} = 0$ and $a_2 = \lrprod{F(e),e'}  = - \lrprod{e,F(e')} = -b_1=:\mu$. Hence 
 \begin{equation}\label{eigeneqs}
  F(e) = \mu e',\quad \quad F(e') = -\mu e,\quad \quad \mu \in \mathbb{R},
 \end{equation}
 which is equivalent to the following eigenequations
 \begin{equation}\label{eigeneq1}
  F(e + i e')= -i\mu (e+i e'),\quad \quad   F(e - i e')= i\mu (e-i e').
 \end{equation}
 In a similar way, for a pair of orthogonal vectors $e_0,e$ spanning a timelike eigenplane, with $e_0$ unit timelike and $e_1$ unit spacelike, one can immediately verify
 \begin{equation}\label{eigeneqstime}
  F(e_0) = \mu e_1,\quad \quad F(e_1) = \mu e_0,\quad \quad \mu \in \mathbb{R}
 \end{equation}
 and 
 \begin{equation}\label{eigeneq2}
  F(e_0 + e_1')= \mu (e_0 +e_1),\quad \quad   F(e_0 - e_1)= -\mu (e_0 - e_1).
 \end{equation}
If $F$ admits an invariant subspace $U$ of any dimension, $F$ also leaves the orthogonal space $U^{\perp}$ invariant. This follows immediately from
% Notice that when there is an $F$-invariant subspace $U$, by skew-symmetry $F$ restrict to both $U$ and its orthogonal because
 \begin{equation}
  0 = \lrprod{F(u),v} = - \lrprod{u,F(v)} \quad \quad \forall u \in U,~\forall v \in U^\perp.
 \end{equation}
In particular, in four dimensions the existence of a timelike eigenplane is equivalent to  the existence of an (orthogonal) spacelike eigenplane.
%  \begin{lemma}
%   $F$ only has real or purely imaginary eigenvalues and $F^2$ only real eigenvalues. 
%  \end{lemma}
% 
% \begin{proof}
%  For en eigenvalue $\lambda$ and eigenvector $v$ of $F^2$, then the conjugate $\lambda^\star, v^\star$ are also an eigenvalue and eigenvector of $F$. The symmetry of $F^2$ leads to $\lambda \lrprod{v,v^\star} = \lambda^\star \lrprod{v, v^\star}$. Then either $\lambda$ is real or $\lrprod{v,v^\star}  =0$. For the latter, $v$ is an eigenvector of $F$ with eigenvalue one of the roots of $\lambda$. Since $F$ is skew, this implies $\lrprod{v,v} =0$, which is only possible if $v = c u$ for some $c \in \mathbb{C},~u \in V$. Then $u$ is an eigenvector
% \end{proof}

Another well-known property of skew-symmetric endomorphisms is that $\dim \mathrm{Im\,} F$ is always even. Equivalently, in four dimensions $\dim \ker F$
is also even (in arbitrary dimension $V$, $\dim \ker F$
  has the same parity as $\dim V$). To see this, consider the 2-form ${\boldsymbol F}$ associated to  $F\in \skwend{V}$ by the standard relation
 \begin{equation}\label{Fflat}
  \boldsymbol{F}(e,e') = \lrprod{e, F(e')},\quad\quad \forall e,e' \in V.
 \end{equation}
 The matrix representing $\boldsymbol F$ is skew in the usual sense. The dimension of $\mathrm{Im\,} \boldsymbol F \subset {V}^\star$ (the dual of $V$) is the rank of this matrix, which is known to be even (see e.g. \cite{gantmacher1960theomat}), and clearly $\dim \mathrm{Im\,} \boldsymbol F =
   \dim \mathrm{Im\,}  F$.

   The first step towards our canonical form for $F$ is the following classification result, which relies on the properties described above.
% This fact is required for the following proof of the classification of $\skwend{\mathbb{M}^{1,3}}$.
 \begin{lemma}[Classification of $\skwend{\mathbb{M}^{1,3}}$]\label{lemmaclasif4}
 Let $F \in \skwend{V}$ in a Lorentzian vector space  $(V,g)$  of dimension four. If $F \neq 0$ then one of the following exclusive possibilities hold:
 \begin{enumerate}
  \item[a)] $F$ has a spacelike eigenvector orthogonal to a null eigenvector, both with vanishing eigenvalue.  
  \item[b)] $F$ has a spacelike eigenplane (as well as a timelike orthogonal eigenplane). 
 \end{enumerate}
\end{lemma}
\begin{proof}
  Since $F$ is not identically zero, $\dim \ker F$ only can be either 2 or 0. Consider first $\dim \ker F = 0$ and let us prove that $b)$ must happen. We show this by proving that
%    The existence of a timelike eigenplane, implies the existence of a spacelike eigenplane and viceversa. In this case
    equations \eqref{eigeneq1} and \eqref{eigeneq2} must be satisfied. 
    Since $\ker F = \lrbrace{0}$, $F$ can only have non-zero eigenvalues, and we already know that they are either  real or purely imaginary. The existence of a purely imaginary one leads to equations \eqref{eigeneq1}, which in turn implies \eqref{eigeneq2}.
    %and $b)$ follows
    Suppose now that all eigenvalues are real non-zero. If there exist two different real eigenvalues $\mu,\mu'$ their respective eigenvectors $w,w'$
   (which recall are null) must satisfy
\begin{equation*}
 \lrprod{F(w),w'} = \mu \lrprod{w,w'} = -\mu' \lrprod{w,w'}.
\end{equation*}
The product $\lrprod{w,w'}$ cannot be zero, as otherwise $w, w'$ would be proportional and the eigenvalues $\mu$ and $\mu'$ would be the same. Thus, 
$\mu = -\mu'$, and hence  \eqref{eigeneq2}, and also $\eqref{eigeneq1}$, hold.
%; or $\lrprod{w,w'} = 0$
%that is not possible because they are null with different eigenvalue.
The remaining case is when all eigenvalues are equal, i.e. the characteristic polynomial is $p_F = (F - I \mu)^4$. By the Cayley-Hamilton theorem $\lrprod{p_f(u),v}= 0, ~\forall u,v \in V$. In particular, $\lrprod{p_f(u),v} = \lrprod{p_f(v),u}, \forall u,v \in V$. By skew-symmetry the even powers on each side cancel out and we are 
left with
\begin{equation}
 - 4 \mu \lrprod{F^3(u),v} - 4 \mu^3 \lrprod{F(u),v} 
 =
 - 4 \mu \lrprod{F^3(v),u} - 4 \mu^3 \lrprod{F(v),u}
  =  4 \mu \lrprod{F^3(u),v} + 4 \mu^3 \lrprod{F(u),v}, \qquad \forall u,v \in V.
\end{equation}
Since we are in the case $\mu \in \mathbb{R} \setminus \{0\}$ we conclude that
$ F(F^2 + \mu^2) = 0$, and since $F$ is invertible ($\ker F = \lrbrace{0}$) also
$F^2 + \mu^2 = 0$. But this means that $F$ admits a complex eigenvalue, which is a contradiction, and we have exhausted all possible cases with $\dim \ker F = 0$.

%  the existence of an $F$-invariant timelike or spacelike plane with $\mu \neq 0$ implies (by simply applying $F$ twice) that $F^2$ has a positive or negative eigenvalue $\pm \mu^2$ respectively. So assume that there are no $F$-invariant planes. Therefore, as a consequence of Lemma \ref{lemmabasics} $e)$, $F^2$ can only have real vanishing eigenvalues. Let $e \in V$ be a non-zero eigenvector $F^2(e) = 0$. Then $F(e) \in \ker F$ must be $F(e) = 0$, so $e \in \ker F$, which is a contradiction. Hence, if $\ker F = \lrbrace{0}$, $F$ must admit an $F$-invariant timelike(spacelike) plane, which implies the $F$-invariance if its orthogonal spacelike (timelike) plane. Moreover, this is the only possibility if $\ker F = \lrbrace{0}$, thus, it excludes any other case.  
 
Now let $\dim \ker F = 2$. According to the causal character of $\ker F$, either $\ker F$ is null, and we are in case $a)$ of the lemma or
$\ker F$ is non-degenerate, and we are in case $b)$. The fact that cases
  $a)$ and $b)$ are mutually exclusive  is obvious.

\end{proof}

The classification in Lemma \ref{lemmaclasif4} contains two possible cases. It is common to use this result to find simple forms for each case, for example, in case $a)$ by including in the basis two orthogonal vectors $k,e \in \ker F$; or in case $b)$, by combining bases in the orthogonal and timelike eigenplanes, so that $F$ is explicitly a direct sum of two 2-dimensional endomorphisms. In the following Proposition we find a canonical form which includes cases $a)$ and $b)$ simultaneously, and which depends on
two parameters only.

\begin{proposition}\label{propcanonF4}
  For every non-zero $F \in \skwend{V}$, with $(V,g)$ a four-dimensional Lorentzian vector space with a choice of time orientation
  , there exists an orthonormal unit basis $B:=\lrbrace{e_0,e_1,e_2,e_3}$, with $e_0$ timelike future directed such that
 \begin{equation}\label{canonFdim4}
   \left ( \begin{array}{c}
             F(e_0) \\
             F(e_1) \\
             F(e_2) \\
             F(e_3)
           \end{array} \right )
                    =  \left(
\begin{array}{cccc}
 0 & 0 & -1 + \frac{\aa }{4} & \frac{\bb }{4} \\
 0 & 0 & 1+\frac{\aa }{4} & \frac{\bb }{4} \\
 -1+\frac{\aa }{4} & - 1- \frac{\aa }{4} & 0 & 0 \\
 \frac{\bb }{4} & -\frac{\bb }{4} & 0 & 0 \\
\end{array}
       \right)
%           =  \left(
%\begin{array}{cccc}
% 0 & 0 & -1 + \frac{\aa }{4} & \frac{\bb }{4} \\
% 0 & 0 & -1-\frac{\aa }{4} & -\frac{\bb }{4} \\
% -1+\frac{\aa }{4} & 1+\frac{\aa }{4} & 0 & 0 \\
% \frac{\bb }{4} & \frac{\bb }{4} & 0 & 0 \\
%\end{array}
       %\right)
       \left ( 
\begin{array}{c}
  e_0 \\
  e_1 \\
  e_2 \\
  e_3
\end{array}
\right )
,\quad\quad \aa, \bb \in \mathbb{R},
 \end{equation}
where $\aa := - \frac{1}{2}\Tr{F^2}$ and $\bb^2 := - 4 \det F$, with $\bb \geq 0$. Moreover, if $\bb=0$ the vector $e_3$ can be taken to be any spacelike unit vector lying in the kernel of $F$.
\end{proposition}

\begin{proof}
  By Lemma \ref{lemmaclasif4} there exist two possible cases. We start proving the proposition assuming that we are in case $a)$. Let $\spn{k,e} = \ker F$, with $k,e \in V$  a pair of orthogonal null and spacelike unit vectors respectively. We can complete them to a semi-null basis $B = \lrbrace{k,l,w,e}$, i.e. such that $\lrprod{k,l}  = -2$, $\lrprod{w,w}=\lrprod{e,e}=1$ and the rest of scalar products all zero.
 Using these orthogonality relations and skew-symmetry of $F$ we can calculate:
 \begin{equation*}
  F(k) = 0, \quad\quad F(l) =  a w,\quad\quad F(w) = \frac{a}{2} k,\quad \quad F(e) = 0,
 \end{equation*}
for a constant $a \in \mathbb{R}\setminus \{0\}$.  Redefine a new basis $\lrbrace{l',k',w',e'}$, with $k' := \frac{\epsilon a}{2} k,~l':= \frac{ 2 \epsilon }{a} l,~w' :=-\epsilon w,~e':=e $, where $\epsilon^2 = 1$ is chosen so that $k',l'$ are future directed. Then
\begin{equation*}
  F(k') = 0, \quad\quad F(l') =   w',\quad\quad F(w') = k',\quad\quad F(e') = 0,
 \end{equation*}
 which in the orthonormal basis  $B = \lrbrace{e_0,e_1,e_2,e_3}$ given by $k' = e_0 + e_1,~l' = e_0-e_1,w' = e_2, e' = e_3$ is
 \begin{equation*}
  F(e_0) = -e_2, \quad\quad F(e_1) = e_2,\quad\quad F(e_2) = -e_0 - e_1,\quad \quad F(e_3)=0.
 \end{equation*}
 This corresponds to expression \eqref{canonFdim4} with $\aa = \bb = 0$.

 It remains to prove the proposition for case $b)$. In this case, there exist timelike and spacelike eigenplanes, $\Pi_t = \spn{e'_0,e'_1}$ and $\Pi_s = \spn{e'_2,e'_3}$ respectively, i.e. fulfilling equations \eqref{eigeneqs} and \eqref{eigeneqstime} for respective eigenvalues $\mu_0$ and $\mu_1$, such that at most one of them vanishes. We can take the bases of $\Pi_t, \Pi_s$ so that that $B' := \lrbrace{e'_0,e'_1,e'_2,e'_3 }$ is an orthonormal
 basis of $V$, with $e_0$ past directed and
 the eigenvalues $\mu_0$ and $\mu_1$ are positive or (at most one) zero . Then, the following change of basis is well-defined:
 \begin{equation}\label{basiscan4}
 \begin{array}{lcl}
 e_0 = \frac{-1}{\sqrt{\mu_0^2 + \mu_1^2}} \lrbrkt{\lr{1 + \frac{\mu_0^2 + \mu_1^2}{4}}e'_0 + \lr{1 - \frac{\mu_0^2 + \mu_1^2}{4}}e'_2}, 
& & 
 e_2  = \frac{1}{\sqrt{\mu_0^2 + \mu_1^2}} \lr{\mu_0 e'_1 + \mu_1 e'_3},\\
 e_1 = \frac{1}{\sqrt{\mu_0^2 + \mu_1^2}} \lrbrkt{\lr{1 - \frac{\mu_0^2 + \mu_1^2}{4}}e'_0 + \lr{1 + \frac{\mu_0^2 + \mu_1^2}{4}}e'_2},
 &  & 
e_3  = \frac{1}{\sqrt{\mu_0^2 + \mu_1^2}} \lr{-\mu_1 e'_1 + \mu_0 e'_3}.
 \end{array}
 \end{equation}
 One checks by explicit computation that  $B:=\lrbrace{e_0,e_1,e_2,e_3}$ is an
   orthonormal basis, with $e_0$ timelike and future directed (because
   $\lrprod{e_0,e'_0}>0$). It is also a matter of direct calculation to see that
 %for the vectors in the basis $B$ it holds
 \begin{equation}\label{syscanonbas4}
 \begin{array}{lcl}
F(e_0)  = \lr{-1 + \frac{\aa}{4}} e_2 + \frac{\bb}{4} e_3, & & F(e_1)   = \lr{1 + \frac{\aa}{4}} e_2 + \frac{\bb}{4}e_3,\\
 F(e_2)  =  \lr{-1 + \frac{\aa}{4}} e_0- \lr{1 + \frac{\aa}{4}} e_1,&  &  F(e_3)  = \frac{\bb}{4}\lr{e_0-e_1},
 \end{array}
 \end{equation}
where the parameters $\aa,\bb \in \mathbb{R}$ are $\aa = \mu_1^2 - \mu_0^2 $ and $\bb = 2 \mu_0 \mu_1 \geq 0$. This corresponds to \eqref{canonFdim4} with at most one of the parameters $\aa, \bb$ vanishing. 

To show the last statement, a simple computation shows that (when
$\bb=0$) the kernel of $F$ is given by
\begin{align*}
  \mbox{ker} F = \left \{ a \left ( 1+ \frac{\aa}{4} \right ) e_0
  + a \left ( 1 -  \frac{\aa}{4} \right ) e_1  + b e_3, \qquad a,b \in \mathbb{R} \right \}.
\end{align*}
The subset of spacelike unit vectors in $\ker F$ is given by $1 + a^2 \aa > 0$ and
$b = \epsilon \sqrt{1 + a^2 \aa}$, $\epsilon = \pm 1$. We introduce the four vectors
\begin{align*}
  e_0' &= \left ( \frac{b + \epsilon}{2} + \left ( 1 + \frac{\aa^2}{16} \right )
  \frac{   b - \epsilon}{\aa} \right ) e_0
  + \left ( 1 - \frac{\aa^2}{16} \right ) \frac{  b- \epsilon }{\aa} e_1
         +   a \left ( 1 + \frac{\aa}{4} \right ) e_3, \\
  e_1' &=  - \left ( 1 - \frac{\aa^2}{16} \right ) \frac{  b-  \epsilon}{\aa} e_0
+ \left ( \frac{ b  +  \epsilon}{2} - \left ( 1 + \frac{\aa^2}{16} \right )
  \frac{ b- \epsilon}{\aa} \right ) e_1
           +  a \left ( - 1 + \frac{\aa}{4} \right ) e_3, \\
  e_2' & = \epsilon e_2, \\
  e_3' & = a \left (1 + \frac{\aa}{4} \right ) e_0
         + a \left ( 1 - \frac{\aa}{4} \right ) e_1 + b e_3,
\end{align*}
and observe that they are well-defined for all values of $\aa$, including zero. A straightforward computation shows that this is an orthonormal basis, and that
%the matrix
%expression of $F$ in this basis is given by
\eqref{canonFdim4} holds with $\bb=0$. The
last statement of the Proposition follows.

\end{proof}

Obtaining a canonical form in the three-dimensional case
is much easier, the main reason being that any two-form in three-dimensions is simple, i.e.   $\bm{F} \wedge \bm{F} = 0$ or, in other words, that $\bm{F}$ is of rank one as a differential form. So, the reader may wonder why it has not been treated before. The reason is that we can obtain the three dimensional case  as a direct corollary of
  the four-dimensional one. The construction is as follows.
Let $F\in \skwend{V}$ with $V$ Lorentzian three-dimensional. 
From $F$ we may define an auxiliary skew-symmetric endomorphism
    $\widehat{F}$ defined on $V \oplus \mathbb{E}_1$ endowed with the product metric ($\mathbb{E}_1$ is the one-dimensional Euclidean space).
    It is obvious that this space is a Lorentzian four-dimensional vector space.
%    space is isometric to $\mathbb{M}^{1,3}$.  
    We denote by $E_3$ a unit vector in $\mathbb{E}_1$  and  define
    $\widehat{F}$ simply by $\widehat{F}(u + a E_3) = F(u) + 0$, for all
    $u \in V$ and $a \in \mathbb{R}$ (we will identify
    $u \in V$ with $u + 0 \in V \oplus \mathbb{E}_1$ from now on). It is immediate to check that $\widehat{F}$ is skew-symmetric. Moreover,
      it has $\bb=0$, by construction. Then, the following Corollary is immediate:

    \begin{corollary}\label{propcanonF3}
        For every non-zero $F \in \skwend{V}$, with $(V,g)$ a Lorentzian three-dimensional vector space with a choice of time orientation, there exists an orthonormal unit basis $B:=\lrbrace{e_0,e_1,e_2}$, with $e_0$ timelike future directed such that
\begin{equation}\label{canonFdim3}
\left ( \begin{array}{c}
          F(e_0) \\
          F(e_1) \\
          F(e_2)
          \end{array} \right )  = \begin{pmatrix}
  0 &  0 & -1 + \frac{\cc}{4} \\
  0 & 0 & - 1 - \frac{\cc}{4} \\
  -1 + \frac{\cc}{4} & 1 + \frac{\cc}{4} & 0
\end{pmatrix}
\left ( \begin{array}{c}
          e_0 \\
          e_1 \\
         e_2
          \end{array} \right ), \quad \quad \cc := -\frac{1}{2}\Tr\lr{F^2} \in \mathbb{R}.
\end{equation}
\end{corollary}
\begin{proof}
  By the last statement of Proposition \ref{propcanonF4}, the canonical basis
 $B = \lrbrace{e_0, e_1, e_2, e_3}$ of
    $\widehat{F}$ can be taken with $e_3 = E_3$, which means that
    $\lrbrace{e_0, e_1, e_2}$ is a basis of $V$.

\end{proof}

% \begin{remark}
%  Notice that what we mean by canonical form is a matrix depending on a certain minimum amount of basis-invariant quantities (trace and determinant), to which every single element $F \in \skwend{V}$ can be carried into. We have not proven that our canonical form is the only one, because indeed there might be more than one. The reason for this requires an analysis that is beyond our purposes here, but roughly speaking, the quantities that define the canonical form are the same than those that define the equivalence classes of the Lorentz group $O(1,n)$ acting on its algebra $o(1,n)$ (represented here by $\skwend{V}$. These are $p = [(n+1)/2]$ quantities and they depend on the entries of $F$, so we just have let $p$ of this entries free and adequately fix the rest. What is special from our canonical form is that the 3-dimensional case is obtained straightforward from the 4-dimensional case by simply making $\beta = 0$. 
% \end{remark}

\begin{remark}[Classification from the canonical form]
  For the canonical forms \eqref{canonFdim3} and \eqref{canonFdim4} we can derive a classification result for skew-symmetric endomorphisms
  and recover Lemma \ref{lemmaclasif4} in terms of $\aa, \bb$. For $F \in \skwend{\mathbb{M}^{1,2}}$ non-zero it is straightforward that $q := (1 + \cc/4) e_0 + (1- \cc/4) e_1$ generates $\ker F$ and furthermore $\lrprod{q, q} = -\cc$. Hence, the sign of $\cc$ determines the causal character of
  the kernel, namely spacelike for $\cc<0$, null for $\cc=0$ and timelike for
    $\cc >0$.  In the four-dimensional case, if
  %Supose now $F \in \skwend{\mathbb{M}^{1,3}}$ non-zero.
   $\bb \neq 0$, then $\ker F = \lrbrace{0}$ and we must be in case $b)$ of Lemma \ref{lemmaclasif4}. If $\bb = 0$, then $e_3 \in \ker F$ (spacelike) and the sign of $\aa$ determines the causal character of a vector $q \in \spn{e_0,e_1,e_2} \cap \ker F$ just like in the previous case. That is, $\bb =0$ and $\aa = 0$ corresponds with case $a)$ of Lemma \ref{lemmaclasif4} and otherwise we are in case $b)$.
\end{remark}

\section{Group of invariance of the canonical form}\label{secinvg}

In this section $F$ is a non-zero skew-symmetric endomorphism in a
  four-dimensional vector space, and
  $B= \{e_0,e_1,e_2,e_3\}$ is a canonical basis, i.e. one where $e_0$
  is future directed and \eqref{canonFdim4} holds.
 %\begin{align*}
%F(e_0) & = \left ( -1 + \frac{\aa}{4} \right ) e_2 + \frac{\bb}{4} e_3, \quad \quad
%F(e_1) = \left ( 1 + \frac{\aa}{4} \right )e_2 + \frac{\bb}{4} e_3, \\
%F(e_2) & = \left ( -1 + \frac{\aa}{4} \right )e_0
%- \left ( 1+ \frac{\aa}{4} \right ) e_1, \quad  \quad
%F(e_3) = \frac{\bb}{4} \left ( e_0 - e_1 \right ),
%\end{align*}
%where $\bb \geq 0$ and $\aa \in \mathbb{R}$.
It is useful to introduce the semi-null basis
$\{ \ell,k, e_2, e_3\}$ defined by
%By definition, we say that $\{ \ell,k, e_2, e_3\}$ is a semi-null basis 
%when $\ell, k$ are both null future directed and
%satisfy $\la \ell, k \ra = -2$ and $\{e_2,e_3\}$ are orthogonal to both
%$\ell,k$, perpendicular to each other and unit.
%We introduce a semi-null basis 
$\ell = e_0 + e_1$, $k= e_0 - e_1$. In this basis the endomorphism $F$ takes the form
\begin{align}
F(\ell) = \frac{\aa}{2} e_2 + \frac{\bb}{2} e_3, \quad\quad 
          F(k) = -2 e_2, \quad\quad          %\label{null1} \\
F(e_2) =- \ell + \frac{\aa}{4} k, \quad\quad 
F(e_3 ) =  \frac{\bb}{4} k. \label{null2}
\end{align}
We are interested in finding the most general orthochronous Lorentz transformation which transforms $B$ into a basis $B' = \{ e_0',e_1',e_2', e_3'\}$ in which $F$ takes the same form.  In terms of the corresponding 
semi-null basis $\{ \ell',k',e_2', e_3'\}$ we must impose \eqref{null2} with  primed vectors. We start with the following lemma:
\begin{lemma}
  \label{lem}
  Let $F$ be skew-symmetric and
  $\{ \ell,k,e_2,e_3\}$ be a semi-null basis that satisfies
  \begin{align}
    F(k) = -2 e_2, \quad \quad F(e_2) =  - \ell + \frac{\aa}{4} k
    \label{first}
  \end{align}
  and
  \begin{align}
    \la F(\ell), F(\ell) \ra
       = \frac{\aa^2 +\bb^2}{4}. \label{second}
  \end{align}
  Then either the semi-null basis $\{ \ell, k, e_2, e_3\}$
  or $\{ \ell, k, e_2, - e_3\}$ fulfils \eqref{null2}, and both do whenever $\bb=0$.
  \end{lemma}

\begin{proof}
   Skew-symmetry imposes $F(\ell)$ and $F(e_3)$ to satisfy
  \begin{align*}
    F(\ell) = \frac{\aa}{2} e_2 + \frac{q}{2} e_3, \quad \quad
    F(e_3) = \frac{q}{4} k', \quad \quad q \in \mathbb{R}.
  \end{align*}
  Condition \eqref{second} imposes $q^2 = \bb^2$. Thus
  $q = \pm \bb$. Since reflecting $e_3$ replaces $q$ by $-q$, either the basis
  $\{ \ell, k, e_2, e_3\}$ or the basis  $\{ \ell, k, e_2, -e_3\}$ satisfies
  \eqref{null2} with $\bb \geq 0$ (and both do in case $\bb=0$). $\hfill \Box$.
  \end{proof}

Thus, to understand the group of invariance of \eqref{null2}
it suffices to impose \eqref{first}-\eqref{second} for $\{\ell',k',e_2'\}$. Let us decompose $k'$ in the original basis as
\begin{align}
  k' = \ca k + \cb \ell + c_2 e_2 + c_3 e_3.
%
%  , %\quad \quad A,B, \cdots = 2,3,
%  \quad \quad 4 ab = ||c||^2
  \label{defkp}
\end{align}
Observe that  $\ca, \cb \geq 0$ as a consequence of
$k'$ being future directed. Let us introduce two vectors $e_2'$ and $\ell'$ so that \eqref{first} are satisfied, namely
\begin{align}
  e_2' & := - \frac{1}{2} F(k') = \left ( \ca - \frac{\cb \aa}{4} \right )
  e_2 - \frac{\cb \bb}{4} e_3 + \frac{c_2}{2} \ell
  - \frac{1}{8} \left ( \aa c_2 + \bb c_3 \right ) k,
  \label{defe2p} \\
  \ell' & := \frac{\aa}{4} k' - F(e_2') = \frac{\cb \left ( \aa^2 + \bb^2 \right )}{16} k
  + \ca \ell - \frac{1}{4} \left ( \aa c_2 + \bb c_3 \right ) e_2
          + \frac{1}{4} \left ( \aa c_3 - \bb c_2 \right ) e_3.
          \label{defellp} 
\end{align}
The conditions  of $k'$ being  null, future directed and $e_2'$  spacelike and unit are easily found to be equivalent
to
  \begin{align}
    - 4 \ca \cb + ||c||^2  & =0, \quad \quad \ca , \cb \geq 0 ,\label{eq1}  \\
 \ca^2 + \frac{\aa^2 + \bb^2 }{16} \cb^2
  + \frac{\aa}{8} \left ( c_2^2 - c_3^2 \right ) +
  \frac{\bb}{4} c_2 c_3 & =1, \label{eq2}
  \end{align}
  where we have set $||c||^2 = c_2^2 + c_3^2$. Under \eqref{eq1}-\eqref{eq2} one
  easily checks that the conditions
 $\la e_2',k'\ra = 0$, $\la e_2',\ell' \ra =0$, 
 $\la \ell', \ell' \ra =0$ and $\la \ell',k'\ra = -2$ are all identically satisfied.  Thus, $\{\ell', k', e_2'\}$ defines a timelike hyperplane and we can introduce $e_3'$ as one of its two unit normals. By construction, the semi-null basis $\{\ell',k', e_2', e_3'\}$
 satisfies \eqref{first}. By Lemma \ref{lem},  this basis or the one
 defined with the reversed $e_3'$ will be a canonical basis of $F$ if and only if
\eqref{second} is satisfied. By skew-symmetry, this condition is equivalent
to
\begin{align}
  \la \ell' , F^2 (\ell') \ra + \frac{\aa^2 + \bb^2}{4} =0.
  \label{cond}
\end{align}
Directly from \eqref{null2} we compute
\begin{align*}
  F^2(\ell) &= - \frac{\aa}{2} \ell  + \frac{\aa^2 + \bb^2}{8} k , \quad
  F^2 (k) = 2 \ell - \frac{\aa}{2} k, \quad
  F^2(e_2)  = - \aa e_2 -  \frac{\bb}{2} e_3, \quad 
  F^2(e_3) = - \frac{\bb}{2} e_2,
  \end{align*}
  from where it follows
  \begin{align*}
    F^2 (\ell') = &
    \frac{1}{2} \left ( \frac{(\aa^2 + \bb^2) \cb}{4} -
    \aa \ca \right ) \ell + \frac{\aa^2 + \bb^2}{8} \left (
                    \ca - \frac{1}{4} \aa \cb \right ) k  + \frac{\left ( 2 \aa^2 + \bb^2
    \right ) c_2 + \aa\bb c_3}{8} e_2
    +\frac{\bb ( \aa c_2 + \bb c_3)}{8} e_3.
  \end{align*}
  One easily checks  that \eqref{cond} is identically satisfied when \eqref{eq1}-\eqref{eq2} hold. Thus, it only remains to solve
  this algebraic system. To that aim, it is convenient to introduce $Q\geq 0$ and an angle $\arg \in [ 0, \frac{\pi}{2}]$ defined by
  \begin{align}
    \aa = Q \cos(2 \arg), \quad \quad \bb = Q \sin (2 \arg).
    \label{Qsigma}
  \end{align}
  When $\aa^2 + \bb^2 >0$, $\{Q, \arg\}$ are uniquely defined. When
  $\aa=\bb=0$, then $Q=0$ and $\arg$ can take any value. Define also
  $\lambda_2, \lambda_3$ by
  \begin{align*}
    c_2 = 2 \lambda_2 \cos \arg  - 2 \lambda_3 \sin \arg, \quad \quad
    c_3 =  2 \lambda_2 \sin \arg  + 2\lambda_3 \cos \arg.
    \end{align*}
    In terms of the new variables, equations \eqref{eq1}-\eqref{eq2} become
    (with obvious meaning for $||\lambda||^2$)
    \begin{align*}
      \ca \cb - ||\lambda||^2 =0, \quad \quad
      16 \ca^2 + Q^2 \cb^2 + 8 Q \left (\lambda_2^2 -\lambda_3^2 \right ) -16 = 0,
      \quad \quad \ca,\cb \geq 0.
    \end{align*}
    When $Q=0$, the solution is clearly $\ca=1, \cb = ||\lambda||^2$, with
    unrestricted $\lambda_2,\lambda_3$. When $Q >0$, we may multiply the first equation by $Q$ and
    find the equivalent problem
    \begin{align*}
      (4\ca + Q \cb)^2= 16 ( 1 + Q \lambda_3^2), \quad \quad 
      (4 \ca - Q \cb)^2 =16 ( 1 - Q \lambda_2^2), \quad \quad
      \ca, \cb \geq 0.
    \end{align*}
    This system is solvable if and only if
    \begin{align}
      \label{ineq}
      |\lambda_2| \leq \frac{1}{\sqrt{Q}}
    \end{align}
    and the solution is given by
        \begin{align}
      \ca = \frac{1}{2} \left ( \sqrt{1 + Q \lambda_3^2}
      + \epsilon \sqrt{1 - Q \lambda_2^2} \right ) , \quad 
\cb = \frac{2}{Q} \left ( \sqrt{1 + Q \lambda_3^2}
      - \epsilon \sqrt{1 - Q \lambda_2^2} \right ), \label{solab}
      \end{align}
      where  $\epsilon = \pm 1$. Observe that the branches $\epsilon=1$ and
      $\epsilon=-1$ are connected to each other across the set
      $|\lambda_2 | = 1/\sqrt{Q}$. Note also that the case $Q=0$ is
      included as a limit $Q \rightarrow 0$ in  the branch  $\epsilon=1$
      (and then the bound \eqref{ineq} becomes vacuous, in accordance with the unrestricted values of  $\{\lambda_2,\lambda_3\}$ when $Q=0$). We can now
      write down explicitly the vectors $\ell',k',e_2'$ defined in \eqref{defkp},
      \eqref{defe2p} and \eqref{defellp}. It is useful to introduce the two spacelike, orthogonal  and unit vectors
            \begin{align*}
        u_2 = \cos \arg \, e_2 + \sin \arg \,e_3, \quad \quad
        u_3 = - \sin \arg \, e_2 + \cos \arg  \, e_3 
      \end{align*}
      which simplify the expression to
      \begin{align*}
        \ell'= & \frac{Q^2}{16} \cb k + \ca \ell + \frac{Q}{2} \left ( -\lambda_2 u_2 + \lambda_3 u_3 \right ),  \\
        k' = & \ca k + \cb \ell + 2 \lambda_2 u_2 + 2 \lambda_3 u_3, \\
             e_2'  = &   \left ( \lambda_2 \cos\arg - \lambda_3 \sin\arg \right )
             \ell - \frac{Q}{4} \left ( \lambda_2 \cos\arg
             + \lambda_3 \sin\arg \right ) k  + \epsilon \cos\arg \sqrt{1 - Q \lambda_2^2} \,  u_2
             - \sin\arg \sqrt{1 + Q \lambda_3^2} \, u_3,
      \end{align*}
      where $\ca,\cb$ must be understood as given by \eqref{solab} (including the limiting case $Q=0$). The fourth vector $e_3'$ is unit and orthogonal
      to all of them. The following pair of vectors satisfy these properties (and of course there are no others).
            \begin{equation}\label{ep3}
            \begin{split}
             e_3' =  \widehat{\epsilon} & \left (
        \left ( \lambda_3 \cos \arg  + \lambda_2 \sin\arg  \right ) \ell +
        \frac{Q}{4}
                 \left ( \lambda_3 \cos\arg  - \lambda_2 \sin\arg  \right ) k  
        + \epsilon \sin\arg \sqrt{1 - Q \lambda_2^2} \,  u_2
        + \cos\arg \sqrt{1 + Q \lambda_3^2} \, u_3 \right )
            \end{split}
      \end{equation}
      where $\widehat{\epsilon} = \pm 1$. It
      is also straightforward to check that $F(e_3') = \widehat{\epsilon} (\bb/4) k'$. Thus, if $\bb \neq 0$, we must choose $\widehat{\epsilon}=1$
      while in the case $\bb =0$ both signs are possible (in accordance with
      Lemma \ref{lem}). Summarizing, the most general orthochronous
      Lorentz transformation that transforms a canonical semi-null basis of $F$ into another one is given by
              \begin{align*}
          \left ( \begin{array}{c}
                  \ell' \\
                  k' \\
                  e_2'\\
                  \widehat{\epsilon} e_3'
                \end{array}
                               \right )
                               = &
                               \left ( \begin{array}{cccc}
\frac{1}{2} \left ( \sqrt{1 + Q \lambda_3^2}
      + \epsilon \sqrt{1 - Q \lambda_2^2} \right ) 
                                         & 
                                           \frac{Q}{8} \left ( \sqrt{1 + Q \lambda_3^2}
      - \epsilon \sqrt{1 - Q \lambda_2^2} \right ) 
                                         &  -Q \lambda_2 /2 & Q \lambda_3/2  \\
                                                 \frac{2}{Q} \left ( \sqrt{1 + Q \lambda_3^2}
      - \epsilon \sqrt{1 - Q \lambda_2^2} \right )  &  \frac{1}{2} \left ( \sqrt{1 + Q \lambda_3^2}
      + \epsilon \sqrt{1 - Q \lambda_2^2} \right )  & 2 \lambda_2  & 2  \lambda_3 \\
           \lambda_2 \cos\arg - \lambda_3 \sin\arg & 
            - Q ( \lambda_2 \cos\arg
             + \lambda_3 \sin\arg )/4  
             &  \epsilon \cos\arg \sqrt{1- Q \lambda_2^2} &
             - \sin\arg \sqrt{1+ Q \lambda_3^2}   \\
 \lambda_3 \cos\arg + \lambda_2 \sin\arg & 
             Q ( \lambda_3 \cos\arg
             - \lambda_2 \sin\arg )/4 
             &  \epsilon \sin\arg \sqrt{1- Q \lambda_2^2} &
                                                            \cos \arg \sqrt{1+ Q \lambda_3^2}
                                       \end{array} \right ) \\
                                     &
                                                              \left (\begin{array}{cccc}
                                        1 & 0 & 0 & 0 \\
                                        0 & 1 & 0 & 0 \\
                                        0 & 0 & \cos\arg & \sin\arg \\
                                        0 & 0 & -\sin\arg & \cos\arg
                                      \end{array}
                                                                \right )
                                                                \left (
                                                                \begin{array}{c}
                                                                  \ell \\
                                                                  k \\
                                                                  e_2 \\
                                                                  e_3                                                                        
                                                                \end{array}
          \right )
          := {\cal T}_F(\lambda_2, \lambda_3, \epsilon)
  \left (                                                              \begin{array}{c}
                                                                  \ell \\
                                                                  k \\
                                                                  e_2 \\
                                                                  e_3                                                                        
                                                                \end{array}
          \right )
      \end{align*}
    where $\widehat{\epsilon} = 1$, unless  $\bb=0$ in which case $
    \widehat{\epsilon}= \pm 1$. Concerning the global structure of the group,
    recall that $\lambda_3$ takes any value in the real line, while
    $|\lambda_2| \leq 1/\sqrt{Q}$. We have already mentioned that as long as $Q \neq 0$, the two branches $\epsilon = \pm 1$ are connected to each other through the values $|\lambda_2|  = 1/\sqrt{Q}$. The topology of the group is therefore
    $\mathbb{R} \times \mathbb{S}^1$ (hence connected) when $Q \neq 0$ and $\bb \neq 0$. When    $Q \neq 0$, $\bb =0$ the group has two connected components (one corresponding to each value of $\widehat{\epsilon}$)
    each one with the topology of $\mathbb{R} \times \mathbb{S}^1$. Finally, when $Q=0$, the group has two connected components (again one for each value of
    $\widehat{\epsilon}$) and the topology of each component is
    $\mathbb{R}^2$.    By construction all elements of the group (in all cases)
    are orthochronous Lorentz transformations. Moreover, it is immediate to check that the determinant of  ${\mathcal T}_F(\lambda_2, \lambda_3, \epsilon)$ is one for all values of $\lambda_2, \lambda_3, \epsilon$. Thus, all elements with $\widehat{\epsilon} = 1$ preserve orientation, while
    the elements with $\widehat{\epsilon} = -1$ reverse orientation.

    \subsection{Invariance group in the three-dimensional case}

    We have found before that for any non-zero skew-symmetric endomorphism $F$
    in $\mathbb{M}^{1,2}$ there exists an orthonormal, future directed basis $B_3=\{e_0,e_1,e_2\}$ where $F$ takes the canonical form \eqref{canonFdim3}. As in the previous case it is natural to ask
    what is the group of invariance of $F$, i.e. the  most general orthochronous Lorentz transformation which transforms $B$ into a basis where $F$ takes the same form.    
    From $F$, recall the auxiliary skew-symmetric endomorphism
    $\widehat{F}$ defined on $\mathbb{M}^{1,2} \oplus \mathbb{E}_1$ that was introduced  before Corollary \ref{propcanonF3}, that is, the endomorphism that acts as $\widehat{F}(u + a e_3) = F(u) + 0$, for all
    $u \in \mathbb{M}^{1,2}$ and $a \in \mathbb{R}$ where $\mathbb{E}_1 = \spn{E_3}$, with $E_3$ unit.     
     Moreover, the basis $B := \{ e_0,e_1,e_2,e_3 = E_3\}$ is canonical for $\widehat{F}$ in the sense of \eqref{canonFdim4} and in addition $\bb =0$. It is clear that there exists
    a bijection between the set
    of orthonormal, future directed bases $B_3' = \{ e_0',e_1', e_2'\}$  where
    $F$ takes its canonical form and the set of future directed
    orthonormal bases $B'$ in $\mathbb{M}^{1,2} \oplus \mathbb{E}_1$ where
    $\widehat{F}$ takes its canonical form  and the last element of $B'$ is
    $E_3$. Thus, in order to determine the group of invariance of
    $F$ it suffices to study the subgroup  of invariance of $\widehat{F}$ which
    preserves the vector $e_3$. Since $\bb=0$ we must impose
    \begin{align*}
      B = Q \sin (2 \arg) =  2 Q \cos \arg \sin\arg =0
    \end{align*}
    and three separate cases arise: (case 1) when $Q\neq 0, \arg=0$, (case 2)
    when $Q=0$ and  (case 3) when $Q\neq 0, \arg = \pi/2$. Equivalently, cases 1, 2 and 3 correspond respectively to $\aa >0$, $\aa=0$ and
    $\aa<0$. Recall also that when $Q=0$ we may choose any value of
    $\arg \in [0,\pi/2]$ w.l.o.g. We choose $\arg=0$ in this case. Recall also that the case $Q=0$ is recovered as a limit $Q \rightarrow 0$ after setting $\epsilon=1$.

    We only need to impose the condition $e_3' = e_3$ in each case. Directly
    from \eqref{ep3} one finds (we also use that $Q = |\aa| $) 
    \begin{align*}
      e_3' & = \widehat{\epsilon} \left ( \lambda_3 \ell + \frac{|\cc|}{4} \lambda_3 k 
      + \sqrt{ 1 + |\cc| \lambda_3^2} \, e_3 \right )
      & & &  \mbox{Case 1} \\
      e_3' & = \widehat{\epsilon} \left ( \lambda_3 \ell + e_3 \right )
      & & & \mbox{Case 2} \\
      e_3' &= \widehat{\epsilon} \left (
      \lambda_2 \ell - \frac{|\cc|}{4} \lambda_2 k + \epsilon  \sqrt{1 - |\cc| \lambda_2^2}
      \right ) \, e_3, & & &  \mbox{Case 3}  
    \end{align*}
    Thus,  cases $1$ and $2$ require $\widehat{\epsilon}=1, \lambda_3 =0$  and
    in case 3 we must set $\widehat{\epsilon} = \epsilon, \lambda_2=0$. Inserting these values in the group of invariance of $\hat{F}$ one finds the most
    general orthochronous Lorentz transformation that preserves the form
    of $F$. We express the result in the canonically associated semi-null
    bases
  %, which we associate canonically to orthonormal, future directed bases by means of} 
%    \begin{align*}
$      \ell = e_0 + e_1, k= e_0 - e_1,  e_2 = e_2$.
%    \end{align*}
    Renaming
    $\lambda_2,\lambda_3$ as $\lambda$, the three cases can be written in the following form
        \begin{align*}
      \left ( \begin{array}{c}
        \ell' \\
                k' \\
        e_2'
      \end{array}
      \right ) & = \left (
      \begin{array}{ccc}
     \frac{1}{2} \left ( 1 + \epsilon \sqrt{1- |\cc| \lambda^2} \right ) & \frac{|\cc|}{8} \left ( 1 - \epsilon \sqrt{1 - |\cc| \lambda^2} \right ) & - \frac{|\cc| \lambda}{2} \\
     \frac{2}{|\cc|} \left ( 1 - \epsilon \sqrt{1 - |\cc| \lambda^2} \right ) & \frac{1}{2} \left ( 1+ \epsilon \sqrt{1 - |\cc| \lambda^2} \right )
     &2 \lambda \\
     \lambda & - \frac{|\cc| \lambda}{4} & \epsilon \sqrt{1 - |\cc| \lambda^2}
      \end{array}
      \right ) \left (
      \begin{array}{c}
        \ell \\
        k \\
        e_2
      \end{array}
      \right ) \quad \cc \geq 0 \\
        \left ( \begin{array}{c}
        \ell' \\
                k' \\
        e_2'
      \end{array}
      \right ) & = \left (
      \begin{array}{ccc}
     \frac{1}{2} \left ( \epsilon + \sqrt{1 + |\cc| \lambda^2} \right ) & \frac{|\cc|}{8} \left ( \sqrt{1 + |\cc| \lambda^2} - \epsilon\right ) & - \frac{|\cc| \lambda}{2} \\
     \frac{2}{|\cc|} \left ( \sqrt{1 + |\cc| \lambda^2} - \epsilon \right ) & \frac{1}{2} \left ( \sqrt{1 + |\cc| \lambda^2} + \epsilon \right ) & 
     - 2 \lambda \\
     - \lambda & - \frac{|\cc| \lambda}{4} &  \sqrt{1 + |\cc| \lambda^2}
      \end{array}
      \right ) \left (
      \begin{array}{c}
        \ell \\
        k \\
        e_2
      \end{array}
      \right ) \quad \cc < 0 
            \end{align*}
    with the understanding that the case $\cc=0$  is obtained from the first expression by setting $\epsilon = 1$ and then performing the limit
    $\cc \rightarrow 0$.

    When $\cc > 0$, the parameter $\lambda$ is restricted
    to $|\lambda | \leq 1/|\cc|$ and the two branches $\epsilon=1$ and
    $\epsilon = -1$ are connected through $ |\lambda| = |\cc|$. The group
    is  connected and  has topology $\mathbb{S}^1$. As an immediate consequence all the elements in the group are not only orthochronous Lorentz transformations (by construction) but also orientation preserving, as they are all connected to the identity. This can also be checked by computing the
    determinant of its matrix representation, which is one irrespectively of
    the value of $\lambda$ and $\epsilon$. When $\cc=0$ the parameter $\lambda$ takes values in the real line and the group has $\mathbb{R}$-topology. Again all its elements are orientation preserving. In fact, in this case the group is simply the set of null rotations preserving $\ell$.
    Finally, in the case $\cc <0$,  $\lambda$ also takes values in the
    real line and the group has two connected components (corresponding to the two values of $\epsilon$). Each component has topology $\mathbb{R}$. The determinant of the matrix representation is now $\epsilon$, so the Lorentz transformations  with $\epsilon=1$ preserve orientation (and define the connected component to the identity) while $\epsilon = -1$ reverse orientation.

    \section{Generators of the invariance group}\label{secgeninvg}

    Returning to the four dimensional case, the identity element $\e$ of the
    group of invariance corresponds to $\lambda_2=\lambda_3=0$ and
    $\epsilon= \widehat{\epsilon}=1$. We may compute the Lie algebra that
    generates it by taking derivatives of the group transformation with respect to $\lambda_2$ and $\lambda_3$ respectively and evaluating at $\e$. This defines two skew-symmetric endomorphisms
    \begin{align*}
      h_2 := \left . \frac{\partial {\mathcal T}_F(\lambda_2,\lambda_3,\epsilon)}{\partial \lambda_2} \right |_{\e}, \quad \quad
      h_3 := \left . \frac{\partial {\mathcal T}_F(\lambda_2,\lambda_3,\epsilon)}{\partial \lambda_3} \right |_{\e}.   
    \end{align*}
 It is immediate to obtain their explicit expression
    \begin{align*}
       \left (
      \begin{array}{l}
        h_2(\ell) \\
        h_2(k) \\
        h_2(e_2) \\
        h_2(e_3)  \end{array}
      \right ) =
      \left ( \begin{array}{cccc}
                0 & 0 & -\frac{Q}{2}  \cos \arg & - \frac{Q}{2} \sin \arg \\
                0 & 0 & 2 \cos \arg & 2 \sin \arg \\
                \cos \arg & - \frac{Q}{4} \cos \arg & 0 & 0 \\
                \sin \arg & - \frac{Q}{4} \sin \arg & 0 & 0
              \end{array}
                                                              \right )
                                                                \left (
      \begin{array}{l}
        \ell \\
        k \\
        e_2 \\
        e_3  \end{array}
      \right ), \\
 \left (
      \begin{array}{l}      
              h_3(\ell) \\
        h_3(k) \\
        h_3(e_2) \\
        h_3(e_3)  \end{array}
      \right ) =
      \left ( \begin{array}{cccc}
                0 & 0 & -\frac{Q}{2} \sin \arg & \frac{Q}{2}  \cos \arg \\
                0 & 0 & -2 \sin \arg & 2 \cos \arg \\
                -\sin \arg & - \frac{Q}{4} \sin \arg & 0 & 0 \\
                \cos \arg & \frac{Q}{4} \cos \arg & 0 & 0
              \end{array}
                                                              \right )
                                                                \left (
      \begin{array}{l}
        \ell \\
        k \\
        e_2 \\
        e_3  \end{array}
      \right ).
    \end{align*}
    Note that any skew-symmetric endomorphism $G$ that commutes with $F$ generates a one-parameter subgroup of Lorentz transformations that leaves the form of $F$ invariant. It follows that this uniparametric
    group is necessarily a subgroup of the full invariance group of $F$.  Hence $G$ must belong to the Lie algebra generated by $h_2$ and $h_3$. Conversely, $h_2, h_3$ (and any linear combination thereof) defines
    a skew-symmetric endomorphism that commutes with $F$.    In other words,  ${\mathcal C}_F := \mbox{span} \{ h_2, h_3 \}$ defines the Lie subalgebra of $so(1,3)$ formed by the elements that commute with $F$. This Lie subalgebra is called the {\it centralizer} of $F$ (e.g. \cite{knapp}) and, as we have just shown, it is two-dimensional for any non-zero $F$. An easy computation shows that $[h_2 , h_3]  =0$, so the centralizer of $F$ is an Abelian Lie algebra. With these properties, it is not difficult to obtain the exponentiated form of the group elements. Define the two $C^1$ functions $t_{\epsilon}(s), t_3(s)$ (prime denotes derivative with respect to $s$)
    \begin{align*}
      t_{\epsilon}^{\prime} & = \epsilon \sqrt{1 - Q t_\epsilon^2}, \quad  \quad \quad \quad t_{\epsilon}(s=0) = 0, \\
      t_3^{\prime} & = \sqrt{1 + Q t_3^2}, \quad  \quad  \quad \quad  t_3(s=0) = 0,
      \end{align*}
      and set
            \begin{align*}
        {\mathcal T}_{\epsilon} (s) :=
        \left (
        \begin{array}{cccc}
\frac{1}{2} \left ( 1 + t_{\epsilon}' \right ) 
                                         & 
                                           \frac{Q}{8} \left ( 1 - t_{\epsilon}' \right )
                                         &  - \frac{Q}{2} \cos\arg  t_{\epsilon} & - \frac{Q}{2} \sin\arg  t_{\epsilon}
                                         \\
                                         \frac{2}{Q} \left ( 1 - t_{\epsilon}' \right )  &  \frac{1}{2} \left ( 1 + t_{\epsilon}' \right )  & 2 \cos\arg  t_{\epsilon} & 2 \sin\arg  t_{\epsilon} \\
          \cos\arg  t_{\epsilon}   & 
            - \frac{Q}{4}\cos\arg  t_{\epsilon} 
                          &  \cos^2\arg t_{\epsilon}' + \sin^2 \arg &
             \sin\arg \cos\arg \left ( t_{\epsilon}'-1 \right )    \\
 \sin\arg  t_{\epsilon} & 
             - \frac{Q}{4} \sin\arg  t_{\epsilon}
             &  \sin\arg \cos\arg \left ( t_{\epsilon}^{\prime} - 1 \right )  & \sin^2\arg t_{\epsilon}^{\prime} + \cos^2 \arg
        \end{array} \right ) \\
        {\mathcal T}_{3} (s) := \left ( \begin{array}{cccc}
\frac{1}{2} \left ( 1 + t_3^{\prime} \right ) 
                                         & 
                                           \frac{Q}{8} \left ( t_{3}^{\prime} -1 \right ) 
                                          &  - \frac{Q}{2} \sin\arg t_{3}
                                          & \frac{Q}{2} \cos\arg t_3
                                             \\
                                                 \frac{2}{Q} \left ( t_3^{\prime} - 1  \right )  &  \frac{1}{2} \left ( 1 + t_3^{\prime} \right )  & - 2 \sin\arg t_3  & 2  \cos\arg t_3  \\
           - \sin\arg t_3 & 
            - \frac{Q}{4} \sin\arg t_3 
                                          &  \cos^2 \arg +t^{\prime}_3 \sin^2 \arg  &    \sin\arg \cos\arg \left ( 1 - t_3^{\prime} \right ) \\    
 \cos\arg t_3  & 
             \frac{Q}{4} \cos\arg t_3 
                                          &  \sin\arg \cos\arg \left (1 - t_3^{\prime} \right ) & \sin^2\arg+ \cos^2\arg t_3^{\prime} 
                                        \end{array} \right )
      \end{align*}
      (in the right-hand sides $t_{\epsilon}, t_{\epsilon}'$ etc. are to be understood evaluated at $s$).  By direct computation one checks that  ($\mbox{Id}$ stands for the $4\times 4$ identity matrix)
      \begin{align*}
&        \frac{d {\mathcal T}_{\epsilon}}{d s}  = h_2 {\mathcal T}_{\epsilon}, \quad \quad {\mathcal T}_{\epsilon=1} (s =0 ) = \mbox{Id}, \\
 &       \frac{d {\mathcal T}_{3}}{d s}  = h_3 {\mathcal T}_{3},
                                         \quad \quad {\mathcal T}_{3} (s=0) = \mbox{Id}, \\
&         \left . {\mathcal T}_{F} (\lambda_2, \lambda_3 , \epsilon) \right |_{\lambda_2 = t_{\epsilon}(s_1), \lambda_3 = t_{3} (s_2)}                                                                   = {\mathcal T}_{\epsilon} (s_1) {\mathcal T}_{3} (s_2)
                                                                  = {\mathcal T}_{3} (s_2) {\mathcal T}_{\epsilon} (s_1).
      \end{align*}
      This shows in particular that ${\mathcal T}_{\epsilon=1}(s) = \mbox{exp}( s h_2)$ and
      ${\mathcal T}_{3} (s) = \mbox{exp} (s h_3)$. Observe also that (in agreement with a previous discussion), when $Q \neq 0$ the branch
      ${\mathcal T}_{\epsilon=-1}$ is connected to the branch 
      ${\mathcal T}_{\epsilon=1}$ because in this case
      \begin{align*}
        t_{\epsilon=1} (s) & = \frac{\sin (\sqrt{Q} s)}{\sqrt{Q}}, \quad
                                       && s \in \left [- \frac{\pi}{2 \sqrt{Q}}, \frac{\pi}{2 \sqrt{Q}} \right ], \\
        t_{\epsilon=-1} (s) & = -\frac{\sin (\sqrt{Q} s)}{\sqrt{Q}}, \quad
                                        &&  s \in \left [- \frac{\pi}{2 \sqrt{Q}}, \frac{\pi}{2\sqrt{Q}} \right ],
      \end{align*}
      so that $s = \pm \pi/(2 \sqrt{Q})$ in the first branch is smoothly
      connected to $s = \mp \pi/(2 \sqrt{Q})$ in the second branch.

         From the matrix representation of $h_2$ and $h_3$ it is obvious (the last two columns are linearly dependent) that $\mbox{det} (h_2)= \mbox{det}(h_3)=0$
so both $h_2, h_3$ are simple, i.e. of matrix rank two. Moreover,
    \begin{align}
      - \mbox{tr} \left ( h_2 \circ h_2  \right ) =  \mbox{tr} \left ( h_3 \circ h_3 \right ) = 2 Q \label{traceh2h3}
    \end{align}
    and $\mbox{tr} \left ( h_2 \circ h_3
      \right ) =0$.    Given that $F$ commutes with itself, i.e. $F \in {\mathcal C}_F$, it must be a linear combination of $h_2$ and $h_3$. Indeed, it is immediate to check that
    \begin{align}
      F = - \cos \arg h_2 + \sin  \arg h_3. \label{lincomb}
    \end{align}
    This expression suggests that the connection between $F$ and the basis $\{h_2,h_3\}$ is via a duality
    rotation. To show that this is indeed the case, we define the one-forms $\{ \bm{\ell}, \bm{k},
    \bm{e_2}, \bm{e_3}\}$ metrically associated to the semi-null basis $\{ \ell, k ,e_2, e_3\}$. Also, for any skew-symmetric endomorphism $F$, we associate the two-form $\bm{F}$ by the standard relation \eqref{Fflat}.
%     \begin{align*}
%       \bm{G} (u,v) := \la u, G(v) \ra, \quad \quad \forall u,v \in \mathbb{M}^{1,3}.
%     \end{align*}
    It is straightforward to find the explicit forms of $\bm{h_2}$ and $\bm{h_3}$ to be\footnote{Our convention for the exterior product is $\bm{u} \wedge \bm{v} := \bm{u} \otimes \bm{v} -
      \bm{v} \otimes  \bm{u}$.}
    \begin{align}
      \bm{h_2} &= \left ( \bm{\ell} - \frac{Q}{4} \bm{k} \right ) \wedge \left (
\cos \arg \bm{e_2} + \sin \arg \bm{e_3} \right ), \label{h2} \\
      \bm{h_3} &= \left ( \bm{\ell} + \frac{Q}{4} \bm{k} \right ) \wedge \left (
      -\sin \arg \bm{e_2} + \cos \arg \bm{e_3} \right ). \nonumber
    \end{align}
    Duality  rotations of a two-form are defined in terms of the Hodge-dual operator, which in turn depends in a choice of orientation in the vector space. To keep the comparison fully general, we let $\kappa = +1$ ($\kappa = -1$) when the orientation in $\mathbb{M}^{1,3}$ is such that the basis $\{ \ell,k, e_2, e_3\}$ is positively (negatively) oriented. Equivalently, if $\bm{\eta}$ is the
    volume form  that defines the orientation, $\kappa$ is given by
    \begin{align}
      \bm{\eta} (\ell,k,e_2,e_3) = 2 \kappa. \label{orient}
    \end{align}
    Let $\bm{G}^{\star}$ denote the Hodge dual associated to $\bm{G}$\footnote{In abstract
      index notation $\bm{G}^{\star}_{\alpha\beta} = \frac{1}{2} \eta_{\alpha\beta\mu\nu} {\bm G}^{\mu\nu}$.}. It is then immediate to check that
    \begin{align*}
      \bm{h_2}^{\star} = \kappa \bm{h_3}.
    \end{align*}
    Defining $\bm{f} := -\bm{h_2}$ and $\mu := -\kappa \arg$, we may rewrite \eqref{lincomb} as
    \begin{align}
      \bm{F} = \cos \mu \bm{f} + \sin \mu \bm{f}^{\star}
      \label{dualrot}
    \end{align}
    which indeed shows that $\bm{F}$ is obtained from the simple form $\bm{f}$ by a duality rotation
    of angle $\mu$. Notice that $f_{\alpha\beta} f^{\alpha\beta} = 2 Q \geq 0$ (by
    \eqref{traceh2h3}). For later use, we observe that
    the most general linear combination $\bm{f} =
    a_0 \bm{h_2} + b_0 \bm{h_3}$ that defines a  simple $2$-form such that
    $f_{\alpha\beta} f^{\alpha\beta} \geq 0$ and \eqref{dualrot} holds for some value of
    $\mu$ is:
    \begin{align}
      Q=0: &\quad \quad \bm{f} = - \cos(\arg+ \kappa \mu) \bm{h_2}
             + \sin ( \arg + \kappa \mu) \bm{h_3},
             \quad \quad & & \mu \in \mathbb{R} \nonumber \\
             Q > 0: & \quad \quad \bm{f} = -\cos(n \pi)  \bm{h_2}, \quad \mu = - \kappa \arg + n \pi, \quad \quad & & n \in \mathbb{N}. \label{Qneq0}                        \end{align}
                        This can be proved easily from the explicit expressions of  $\bm{h_2}, \bm{h_3}$ and the fact that they are linearly independent simple $2$-forms.

                        One may wonder whether this connection with duality rotations could have been used as the starting point to obtain in an easy and natural way the canonical form of $F$. We will argue that this alternative approach, although possible, it is far from obvious and cannot be regarded as natural.

    We fix a skew-symmetric endomorphism $F$ in a a four-dimensional vector space with a Lorentzian metric, and let $\bm{F}$ be the metrically associated $2$-form. Define as before $\aa :=
    - \frac{1}{2} \mbox{trace} \left ( F^2 \right )$ and $\bb^2 = - 4 \det(F)$,
    $\bb >0$ where the determinant is taken for any matrix representation of $F$ in an orthonormal basis. The invariants $\aa$ and $\bb$ are directly related to the two algebraic invariants of $\bm{F}$ as
    \begin{align}
      \aa = \frac{1}{2} F_{\alpha\beta} F^{\alpha\beta}, \quad \quad \bb = \frac{1}{2} \mbox{abs}\left (  F_{\alpha\beta} F^{\star}{}^{\alpha\beta} \right ).
      \label{invs}
    \end{align}
    The first one follows trivially from the definition of $\aa$. The second is a well-known algebraic identity that can be found e.g. in \cite{LichnerowiczTheoRelGravEM}. Given $F$, a duality rotation of angle $-\mu$ defines the $2$-form $\Fmu$ as \cite{Rainich1924},
    \cite{MisnerWheeler1957},
    \begin{align}
      \Fmu := \cos \mu \bm{F} - \sin \mu \bm{F}^*. \label{Duality}
    \end{align}
   A simple computation shows that $\Fmu$ is simple  (i.e. $\Fmuind{}_{\alpha\beta} \Fmuind{}^{\star}{}^{\alpha\beta} =0$) and
    satisfies $\Fmuind{}_{\alpha\beta} \Fmuind{}^{\alpha\beta} \geq 0$
    if and only if (cf. \cite{MisnerWheeler1957})
    \begin{align}
      \aa \sin (2\mu) + \widehat{\kappa} \bb \cos(2 \mu)  & =0, \nonumber \\
      \aa \cos (2 \mu) - \widehat{\kappa} \bb \sin (2\mu) & \geq 0, \label{conds}
          \end{align}
    where $\widehat{\kappa}$ is the sign defined by $\frac{1}{2} F_{\alpha\beta} F^{\star} {}^{\alpha\beta}
    = \widehat{\kappa} \bb$ (when $\bb=0$, $\widehat{\kappa}$ can take any value
    $\widehat{\kappa} = \pm 1$). Inserting \eqref{Qsigma} we find that whenever $Q=0$ all values of $\mu$ solve \eqref{conds} (which reflects the fact that $\bm{F}$ is null, and so are all its duality rotated $2$-forms). When $Q \neq 0$, the solutions of \eqref{conds} are $\mu = - \widehat{\kappa} \arg + n \pi$, $n \in \mathbb{N}$.  Thus, we recover the expression in \eqref{Qneq0} provided we can ensure that $\widehat{\kappa} = \kappa$. Note that the sign of $F_{\alpha\beta} F^{\star}{}^{\alpha\beta}$ only depends on $F$ and the choice of orientation. It is a matter of direct checking that $F$ as given in \eqref{null2} with the choice of orientation where \eqref{orient} holds  satisfies $F_{\alpha\beta} F^{\star}{}^{\alpha\beta} = 2 \kappa \bb$, so that indeed $\widehat{\kappa} = \kappa$ follows (unless $\bb=0$, of course, in which  case $\widehat{\kappa} = \pm 1$).

    We can now show how the canonical basis can be constructed from $F$ using a duality rotation approach. Fixed an orientation on the vector space (i.e. a choice of volume form \bm{\eta}, and its associated Hodge dual) define $\aa$ and
    $\bb$ as in \eqref{invs}. Let $\widehat{\kappa} \in \{ -1,1\}$ be such that  $2 \widehat{\kappa} = F_{\alpha\beta} F^{\star}{}^{\alpha\beta}$ (if $\bb=0$, we allow any sign for
    $\widehat{\kappa}$). Introduce  $\arg$ so that \eqref{Qsigma} holds with
    $\arg \in [0, \pi/2]$ (if $\aa=\bb=0$ then $\arg$ can take any value in this interval). Define then $\mu = - \widehat{\kappa} \arg$ and construct
    $\Fmu$ by \eqref{Duality}. We let $\bm{h_2} := - \Fmu$. Since this $2$-form is simple, there exist two linearly independent vectors $a, b$ such that $\bm{h_2} = \bm{a} \wedge \bm{b}$. These vectors are obviously not unique, but certainly at least one of them must be spacelike. It can also be taken unit. We let $E_2 := b$ have this property. Exploiting the freedom $a \rightarrow a + s E_2$, $s\in \mathbb{R}$ we may take $a$ perpendicular to $E_2$. By construction $(h_2){}_{\alpha\beta} (h_2)^{\alpha\beta}  \geq 0$ (recall \eqref{conds}) which is equivalent to $\la a, a \ra \geq 0$, i.e. $a$ is spacelike or null. Let $Q \geq 0$ be defined by $Q = \la a, a \ra$. It is clear that there exists a timelike plane $\Pi$ containing $a$ and orthogonal
    to $E_2$ (this plane is obviously non-unique). Fixed $\Pi$, it is easy to show that
    there exists a future directed a null basis  $\{ \ell, k\}$ on $\Pi$ satisfying  $\la \ell, k \ra = -2$ and  such that $\bm{a} = \bm{\ell} - (1/4) Q  \bm{k}$. Finally, consider the timelike hyperplane
    defined by $\mbox{span} \{ \ell, k, E_2 \}$ and select the unique unit normal $E_3$ to this hyperplane satisfying the orientation requirement (cf. \eqref{orient})
    \begin{align*}
      \eta(\ell,k, E_2, E_3 ) = 2 \widehat{\kappa}.
    \end{align*}
    So far, from a non-zero $F$  we have constructed 
      a (collection of) semi-null basis $\{\ell,k,E_2,E_3\}$ in quite a natural way. Observe that when $\aa=\bb=0$, the angle $\arg$ is arbitrary, so the semi-null basis has extra additional freedom in this case. What appears to be hard to guess from this
    construction is that instead of $\{E_2, E_3\}$ we should introduce $\{e_2, e_3\}$ by means of  the   $\arg$-dependent rotation (cf. \eqref{h2}) 
    \begin{align}
      E_2 = \cos \arg e_2 + \sin \arg e_3, \quad \quad
      E_3 = - \sin \arg e_2 + \cos \arg e_3. \label{rot}
    \end{align}
    It is by using this transformation that the form of $F$ in the basis
    $\{ \ell,k,e_2, e_3\}$ takes a form that depends \underline{only}
    on the invariants $\aa, \bb$. It is remarkable that the
    $\arg$-freedom inherent to the case $\aa= \bb=0$ (i.e. when $F$ is null) drops out after performing the rotation \eqref{rot}, and we get
    a canonical form  that covers all cases and depends only on $\aa$ and $\bb$, irrespectively of which values these invariants may take.

\section{Global conformal Killing vectors on the plane}\label{secGCKVs}

In the following sections we connect our previous results with the Lie algebra of conformal Killing vector fields of the sphere and the group of motions they generate, i.e. the  M\"obius group. In our analysis, it is useful to employ the Riemann sphere $\mathbb{C} \cup \lrbrace{\infty}$. Although we will rederive some of the results we need here, we refer the reader to \cite{needham1998visual} and \cite{schwerdtgeomCnum} for more details about the M\"obius transformations on the Riemann sphere. Some of the contents may also be found in other more general references such as \cite{PenroseRindVol1} and \cite{IntroCFTschBook}.  Regarding Lie groups and Lie algebras, most of the results we will employ can be found in introductory level textbooks such as \cite{hall2003lie}, but other references \cite{goodmanwallach}, \cite{knapp} are also appropriate.

Consider the euclidean plane $\mathbb{E}^2 = (\mathbb{R}^2, g_E)$ and select Cartesian
coordinates $\{ x,y\}$. It is well-known that the set of conformal Killing vectors (CKV) on $\mathbb{E}^2$ is given by 
\begin{align*}
  \xi = U(x,y) \partial_x + V(x,y) \partial_y
\end{align*}
where $U,V$ satisfy the Cauchy-Riemann conditions
$\partial_x U = \partial_y V$, $\partial_y U = - \partial_x U$. These vector fields satisfy
\begin{align*}
  \pounds_{\xi} g_E = 2 \left ( \partial_x U + \partial_y V \right ) g_E.
\end{align*}
Consider the one-point compactification of $\mathbb{E}^2$ into the Riemann sphere $\mathbb{S}^2$. It is also standard that that set of conformal Killing vectors
that extend smoothly to $\mathbb{S}^2$ is given by the subset of CKV for which $U$ and $V$ are polynomials of degree at most two. We name them {\it global
  conformal Killing vectors} (GCKV). Thus, the set of GCKV is parametrized by six real constants $\{ b_x, b_y, \nu, \omega, a_x, a_y \}$  and take the form
\begin{align}
  \xi & =
  \left ( b_x  + \nu x - \omega y + \frac{1}{2} a_x \left ( x^2 - y^2
  \right ) + a_y xy \right ) \partial_x
  + \left ( b_y  + \nu y + \omega x + \frac{1}{2} a_y \left ( y^2 - x^2
  \right ) + a_x xy \right ) \partial_x \\
      & = \xi(a_x,b_x, \nu, \omega, b_x, b_y )
        \label{xiexp}
\end{align}
It is clear that the use of complex coordinates is advantageous in this
context. For reasons that will be clear later, it is convenient for us to
introduce the complex
coordinate $\z = \frac{1}{2} (x- i y )$. In terms of $\z$, the set of CKV is given by $\xi  = f \partial_\z+ \overline{f} \partial_{\overline{\z}}$ (recall that  bar denotes complex conjugation) where $f$ is a holomorphic function of $\z$, while $U,V$ are defined by $2 f =  U - i V$. The set of GCKV is parametrized by three complex constants
$\{ \mu_0, \mu_1, \mu_2 \}$ as
\begin{align}
  \xi = \left ( \mu_0 + \mu_1 \z + \frac{1}{2} \mu_2 \z^2  \right ) \partial_\z
  + \left ( \overline{\mu_0} + \overline{\mu_1} \overline{\z} + \frac{1}{2} \overline{\mu_2} \overline{\z}^2  \right ) \partial_{\overline{\z}}.
  \label{muform}
\end{align}
The relationship between the two sets of parameters is immediately checked to be (we emphasize that  this specific form depends on our choice of complex coordinate $\z$)
\begin{align}
  \mu_0 = \frac{1}{2} \left (b_x - i b_y \right ), \quad
  \mu_1 = \nu - i \omega, \quad \mu_2 = 2 \left ( a_x + i a_y \right ).
  \label{mudata}
\end{align}
We  denote the GCKV with parameters $\mu := (\mu_0,\mu_1, \mu_2)$ as $\xi_{\{ \mu\} }$. We shall need the following lemma concerning orthogonal and commuting GCKV. The result should be known but we did not find an appropriate reference.
\begin{lemma}
  \label{orto}
  Let $\xi_{\{\mu\}}, \xi_{\{\sigma\}}$ be global conformal Killing vector fields on
    $\mathbb{E}^2$ with corresponding parameters
    $\mu = \{  \mu_0, \mu_1, \mu_2 \},
  \sigma = \{  \sigma_0, \sigma_1, \sigma_2 \}$. Assume that
  $\xi_{\{ \mu\}}$ is not the zero vector field. Then
  \begin{enumerate}
    \item 
    $\xi_{\{\sigma\}}$ is everywhere
    perpendicular to $\xi_{\{\mu\}}$
    if and only if
    $ \sigma = i \, r \, \mu$ with $r \in \mathbb{R}$.
  \item $\xi_{\{\sigma\}}$ commutes with $\xi_{\{\mu\}}$ if and only if
    $ \sigma = c \mu$ with $c \in \mathbb{C}$.
    \end{enumerate}
    Moreover,
    $\xi_{c \mu}$ has Euclidean  norm
\begin{align*}
  g_E (\xi_{ \{ c  \mu \}}, \xi_{ \{ c  \mu \}} ) |_p =
  |c|^2  g_E (\xi_{ \{ \mu \}}, \xi_{ \{ \mu \}} ) |_p, \quad \quad \forall p \in \mathbb{E}^2.
  \end{align*}
\end{lemma}
\begin{proof}
  Let $f_{\mu} = \mu_0 + \mu_1 z + \frac{1}{2} \mu_2 \z^2$ so that
$\xi_{\{ \mu \}} = f_{\mu} \partial_\z + \overline{f_{\mu}} \partial_{\overline{\z}}$ and define  $f_{\sigma}$ correspondingly. The euclidean metric is $g_E = 4 d \z d \overline{\z}$, so
\begin{align}
  g_E (\xi_{\{\mu\}}, \xi_{\{ \mu\}}) |_p = 2 \left ( f_{\mu} \overline{f_{\sigma}} + \overline{f_{\mu}} f_{\sigma} \right ) |_{z(p)}.
  \label{norm}
\end{align}
The condition of orthogonality is equivalent to $f_{\mu} \overline{f_{\sigma}}
+ \overline{f_{\mu}} f_{\sigma} =0$. This is a polynomial in $\{\z, \overline{z}\}$, so its vanishing is equivalent to the vanishing of all its coefficients. Expanding, we find
\begin{align}
  \mu_0 \overline{\sigma_0} + \overline{\mu_0} \sigma_0 =0,   \quad \quad
  \mu_1 \overline{\sigma_1} + \overline{\mu_1} \sigma_1 =0,   \quad \quad
  \mu_2 \overline{\sigma_2} + \overline{\mu_2} \sigma_2 =0,  \label{firstline} \\
  \mu_1 \overline{\sigma_0} + \overline{\mu_0} \sigma_1 =0, \quad \quad
  \mu_2 \overline{\sigma_0} + \overline{\mu_0} \sigma_2 =0, \quad \quad
  \mu_2 \overline{\sigma_1} + \overline{\mu_1} \sigma_2 =0. \label{secondline}
\end{align}
Equations \eqref{firstline} are equivalent to the existence of three real numbers
$\{q_1, q_2,q_3\}$ such that $\mu_a \overline{\sigma_{a}} = i q_a$, $a=0,1,2$.
Multiplying the equations in \eqref{secondline} respectively by $\mu_0 \overline{\mu_1}$, $\mu_0 \overline{\mu_2}$ and $\mu_1 \overline{\mu_2}$ one finds
\begin{align*}
  & q_0 |\mu_1|^2 - q_1 |\mu_0|^2 =0, \quad
  q_0 |\mu_2|^2 - q_2 |\mu_0|^2 =0, \quad
  q_1 |\mu_2|^2 - q_2 |\mu_1|^2 =0 \quad   \Longleftrightarrow \quad 
  \left ( q_0, q_1,q_2 \right ) \times \left (|\mu_0|^2, |\mu_1|^2, |\mu_2|^2 \right ) =
  (0,0,0),
\end{align*}
where $\times$ stands for the standard cross product.
Since $( |\mu_0|^2, |\mu_1|^2, |\mu_2|^2 ) \neq (0,0,0)$ (from our assumption that
$\xi_{\{\mu\}}$ is not identically zero) there exists a real number $r$ such that
$(q_0,q_1,q_2) = -r (|\mu_0|^2, |\mu_1|^2, |\mu_2|^2 )$. Thus
$\mu_a \overline{\sigma_a} = - i r |\mu_a |^2$. Fix $a \in \{ 0,1,2\}$. If $\mu_a \neq 0$, it follows
that $\overline{\sigma_a} = - i r \overline{\mu_a}$. If, instead, $\mu_a=0$
then  it follows from \eqref{secondline} (since at least of the $\mu$'s is not zero)  that $\sigma_a =0$. In either case
we have $\sigma_a = i r \mu_a$. This proves point 1. in the lemma.

For point 2. we compute the Lie bracket and find
\begin{align*}
  \left [ \xi_{\{\mu\}}, \xi_{\{\sigma\}} \right ]
  =
  \left ( f_{\mu} \frac{d f_{\sigma}}{d \z} -
  f_{\sigma} \frac{d f_{\mu}}{d \z} \right ) \partial_{\z}
  +
    \left ( \overline{f_{\mu}} \frac{d \overline{f_{\sigma}}}{d \overline{\z}} -
  \overline{f_{\sigma}} \frac{d \overline{f_{\mu}}}{d \overline{\z}} \right ) \partial_{\overline{\z}}.
\end{align*}
The two vectors commute iff
\begin{align*}
f_{\mu} \frac{d f_{\sigma}}{d \z} -
f_{\sigma} \frac{d f_{\mu}}{d \z} & =
\mu_0 \sigma_1 - \mu_1 \sigma_0 + \left ( \mu_0 \sigma_2 - \mu_2 \sigma_0 \right )\z + \frac{1}{2} \left (  \mu_1 \sigma_2 - \mu_2 \sigma_2 \right ) \z^2 =0 \\
\Longleftrightarrow & \quad \quad
\left (\sigma_0, \sigma_1, \sigma_2 \right )  \propto (\mu_0,\mu_1,\mu_2), 
\end{align*}
and point 2. is proved. The last claim of the lemma follows from \eqref{norm} and
the linearity $f_{c \mu} = c f_{\mu}$. 
\end{proof}

An immediate corollary of this result is that the set of GCKV that commute with a given GCKV $\xi_{\{\mu\}}$ is two-dimensional and generated by
$\xi_{\{\mu\}}$ and $\xi_{\{\mu\}}^{\perp}:= \xi_{\{- i \mu\}}$.

Recall that a M\"obius transformation is a diffeomorphism of the Riemann sphere $\mathbb{C} \cup \{ \infty \}$ of the form
\begin{align}
  \chi^{\A}:  \mathbb{C} \cup \{ \infty \} & \longrightarrow \mathbb{C} \cup
  \{ \infty \} \nonumber \\
  z & \longrightarrow \chi^{\A} (z) = \frac{\alpha \z + \beta}{\gamma \z +
    \delta}, \quad \quad \A := \left ( \begin{array}{cc}
    \alpha & \beta  \\
    \gamma & \delta
  \end{array}
             \right ), \quad \quad \alpha \delta - \beta \gamma = 1.
             \label{matrixA}
\end{align}
The set of M\"obius transformations forms a group under composition, which we denote by $\Mo$, and the map $\chi : SL(2, \mathbb{C}) \longrightarrow \Mo$ defined by $\chi (\A) = \chi^{\A}$ is a group morphism. The kernel  
of this morphism is $K:= \{ \mathbb{I}_2, - \mathbb{I}_2\}$ and in fact
$\chi$ descends to an isomorphism between $PSL(2, \mathbb{C}) := SL(2, \mathbb{C})/K$ and $\Mo$. 
In geometric terms, the M\"obius group corresponds to the set of orientation-preserving conformal diffeomorphisms of the standard sphere $(\mathbb{S}^2, g_{\mathbb{S}^2})$
(recall that
a diffeomorphism $\chi := \mathbb{S}^2 \longrightarrow \mathbb{S}^2$ is
conformal if $\chi^{\star} (g_{\mathbb{S}^2})  = \Omega^2 g_{\mathbb{S}^2}$ for some
$\Omega \in C^{\infty} ( \mathbb{S}^2, \mathbb{R}^+)$). The M\"obius group thus transforms conformal Killing vectors of $\mathbb{S}^2$ into themselves, and, hence
it also transforms global GCKV of $\mathbb{E}^2$ into themselves. In other
words, given a GCKV $\xi_{\{ \mu \}}$, the vector field  $\chi^{\A}_{\star}
(\xi_{\{\mu\}})$ is also a GCKV\footnote{Note that $\chi^{\A}$ has singularities as a map
from $\mathbb{E}^2$ into $\mathbb{E}^2$, but $\chi^{\A}_{\star} (\xi_{\{\mu\}})$
extends smoothly to all $\mathbb{E}^2$, and in fact to the whole Riemann sphere. Again this is standard and well-understood, so we will abuse the notation and write $\chi^{A}_{\star}$ as if the map $\chi^A$ were well-defined everywhere on
$\mathbb{E}^2$}. Let  $\mu^{\prime} := (\mu^{\prime}_0, \mu^{\prime}_1,
\mu^{\prime}_2 )$ be the set of parameters of $\chi_{\star}^{\A} (\xi_{\{ \mu\}})=: \xi_{\{\mu'\}}$. A
straightforward computations shows that
\begin{align}
  \left ( \begin{array}{c}
    \mu^{\prime}_0 \\
    \mu^{\prime}_1 \\
    \mu^{\prime}_2
  \end{array} \right )
  = 
  \underbrace{ \left (
    \begin{array}{ccc}
    \alpha^2  & - \alpha \beta & \frac{1}{2} \beta^2 \\
    - 2 \alpha \gamma & \alpha \delta + \beta \gamma & -\beta \delta \\
    2 \gamma^2 & -2 \gamma \delta & \delta^2
  \end{array}
  \right )}_{:=\Q_{\A}}
\left ( \begin{array}{c}
    \mu_0 \\
    \mu_1 \\
    \mu_2
\end{array} \right ).
\label{transmu}
\end{align}
The determinant of this matrix is one, so $\Q_{\A} \in SL(3, \mathbb{C})$.
As a consequence of  $\chi^{\A_1} \circ \chi^{\A_2} = \chi^{\A_1 \cdot \A_2}$ (where
$\cdot$ denotes product of matrices), it follows that the map
$\Q : SL(2, \mathbb{C} ) \longrightarrow SL(3,\mathbb{C})$ defined by
$\Q(\A) = \Q_{\A}$ is a morphism of groups, i.e.
$\Q_{\A_1} \cdot \Q_{\A_2} = \Q_{\A_1 \cdot  \A_2 }$. This property can also be confirmed by explicit computation. In particular $\Q$ defines a representation of the group $SL(2,\mathbb{C})$ on $\mathbb{C}^3$. It is easy to show that this representation is actually isomorphic to the adjoint representation. Recall that
for matrix Lie group $G$ (i.e. a Lie subgroup of $GL(n,\mathbb{C})$), the
adjoint representation $\mbox{Ad}$ takes the explicit form (e.g. \cite{hall2003lie})
\begin{align*}
  \mbox{Ad}: G & \longrightarrow \mbox{Aut} (\g) \\
  g  &\longrightarrow Ad(g):=Ad_g : \begin{array}{lll}
                                      & & \\
                                      \g & \rightarrow & \g \\
                                     X  & \rightarrow & g X g^{-1}
                                          \end{array}
                                   \end{align*}
where $\g$ is the Lie algebra of $G$ and $\mbox{Aut} (\g)$
is the set of automorphisms of $\g$. The isomorphism between $\Q$ and
$\mbox{Ad}$ is as follows.  Let us choose the  basis of $sl(2,\mathbb{C})$ given by
\begin{align*}
  \w^0 := \left ( \begin{array}{ll}
           0 & 2 \\
           0 & 0
         \end{array} \right )
               \quad \quad
  \w^1 := \left (\begin{array}{ll}
           1 & 0 \\
           0 & -1
                 \end{array} \right )
               \quad \quad
                 \w^2 := \left ( \begin{array}{ll}
           0 & 0 \\
           -1 & 0
         \end{array} \right )
\end{align*}
and define the vector space isomorphism $h : \mathbb{C}^3 \rightarrow sl(2,\mathbb{C})$ defined by $h(\mu_0, \mu_1, \mu_2) = \mu_a \w^a$ ($a,b, \cdots = 0,1,2$). One then checks easily by  explicit computation that
$h^{-1} \circ Ad_g  \circ h = \Q(g)$, for all $g \in SL(2,\mathbb{C})$.

Recall that the Killing form of a Lie algebra $\g$ is the symmetric bilinear map on $\g$ defined by
$B (\a_1, \a_2) :=\Tr \left ( \mbox{ad}(\a_1) \circ \mbox{ad}(\a_2) \right )$
where $ad(\a)$, $\a \in \g$ is the adjoint endomorphism $ad(\a) : \g \rightarrow \g$
defined by $ad(\a)(\b) := [\a,\b]$. The Lie algebra $sl(2,\mathbb{C})$ is semi-simple, so its Killing form is non-degenerate (e.g. \cite{knapp}). The explicit form in the basis $\{ \w_0, \w_1, \w_2\}$ is given by
\begin{align*}
  B(\mu_a \w^a, \sigma_a \w^a ) = 8 \left ( \mu_1 \sigma_1 - \mu_0 \sigma_2 - \mu_2 \sigma_0 \right ).
\end{align*}
A fundamental property of the Killing form is that it is invariant under
automorphisms (see e.g. \cite{cap}), so in particular under the
adjoint representation $B(\mbox{Ad}_g(\a), \mbox{Ad}_g(\b)) = B(\a,\b)$ 
for all
$g \in G$.  Given $\{ \mu\}$ we define two real quantities $\aa_{\{\mu\})
 }$, $\bb_{\{\mu\}}$ by
  \begin{align*}
   \aa_{\{\mu\}} - i \bb_{\{\mu\}} := 2 \mu_0 \mu_2 - \mu_1^2.
 \end{align*}
  As a consequence of the discussion above, the quantities  $\aa_{\{\mu\}}$, $\bb_{\{\mu\}}$
    associated to a GCKV $\xi_{\{\mu\}}$ are invariant under M\"obius transformations. We have now all necessary ingredients to determine the set of
    M\"obius transformations that transform a GCKV into its canonical form.
    Before doing so, however, we summarize known results on the relationship between GCKV
    and skew-symmetric endomorphism in the Minkowski spacetime.

    \section{GCKV and skew-symmetric endomorphisms}
    \label{GCKV_and_F}

    It is well-known that conformal diffeomorphisms 
      on the
      standard sphere of dimension $n \geq 2$, $\mathbb{S}^n$,
  %    (Recall that a conformal transformation of $(M,g)$ is a diffeomorphism $\Phi$ of $M$ satisfying $\Psi^{\star} (g )=
  %  \Omega^2  g$ for some smooth positive function $\\Omega : M \rightarrow 
  %  \mathbb{R}$ depending of $\Phi$.} 
  are in one-to-one correspondence with
  orthochronous Lorentz transformations in the
  Minkowski spacetime $\mathbb{M}^{1,n+1}$. The underlying reason (see
  e.g. \cite{PenroseRindVol1} or \cite{IntroCFTschBook}) is that such Lorentz transformations leave   invariant the future null cone, and the set of null semi-lines in the cone
  admits a differentiable structure and a metric that makes it isometric
  to $\mathbb{S}^n$. The action of the orthochronous Lorentz group on the set of future directed null semi-lines gives rise to a conformal transformation, defining a map  that turns out to be one-to-one. This property translates, at the infinitesimal level, to the existence of a one-to-one map between
  conformal Killing vectors of $\mathbb{S}^n$  and the set
  of skew-symmetric endomorphisms in $\mathbb{M}^{1,n+1}$. The explicit form of these two maps depends on the choice of isometry between the set of null-semilines and
  $\mathbb{S}^n$. This freedom amounts, essentially
  to fixing
  a future directed orthonormal Lorentz frame $\{ e_{\alpha}\}$ with associated
  Minkowskian coordinates $\{ T, X^i\}$ 
  in $\mathbb{M}^{1,n+1}$ and selecting a unit spacelike direction $u = u^i e_i$
  with respect to which one performs a stereographic projection of the sphere
  $\{T =1, \sum_{i=1}^{n+1}  (X^i)^2 = 1\}$ minus the point $p_u := \{X^i = u^i\}$
  onto an $n$-dimensional spacelike plane $\Pi_u$ that lies in
    the hyperplane  $\{T = 1\}$, is orthogonal to $u$ and does not contain
  the point $p_u$ (such a plane is uniquely defined by the signed euclidean distance from $\Pi_u$ and $p_u$ in the Euclidean plane $\{ T = 1\}$). The final choice is a set of Cartesian coordinates in $\Pi_u$.

The construction above can also be done using the hyperboloid of timelike unit future vectors $\mathcal{H} \subset \mathbb{M}^{1,n+1}$, whose isometries are the orthochronous Lorentz transformations. The boundary $\partial \mathcal{H}$ of the conformal compactification of the hyperboloid (which represents ``infinity'' of  $\mathcal{H}$)  is a standard sphere, where the action of the Lorentz group can be extended and it turns out to generate conformal transformations.
Details of this construction can be found e.g. in Appendix A
of \cite{Kdslike}. As in the other representation, the details of the map depend on
  how the sphere at infinity is introduced. The way how the explicit construction was carried
  out in \cite{Kdslike} corresponds, in the  description above, to choosing the vector $u = -e_1$, the plane $\Pi_u = \{ T =1, X^1 = 1\}$ and
  Cartesian coordinates in $\Pi_u$ given by $\{ X^2, \cdots, X^{n+1}\}$.
  With these choices, and restricting to dimension $n=2$,  

  the explicit map between the set
  of skew-symmetric endomorphisms  $\skwend{\mathbb{M}^{1,3}}$ and the set of Global Conformal
  Killings vectors on $\mathbb{R}^2$ (denoted by $\mbox{GCKV}(\mathbb{R}^2)$) is
    \begin{align}
      \Psi := \skwend{\mathbb{M}^{1,3}} & \longrightarrow  \mbox{GCKV} (\mathbb{R}^2) \nonumber \\
      F = \left (  \begin{array}{cccc}
                     0 & -\nu & - a_x + \frac{b_x}{2} & - a_y + \frac{b_y}{2} \\
                     - \nu & 0 & - a_x - \frac{b_x}{2} & - a_y - \frac{b_y}{2} \\
                     - a_x + \frac{b_x}{2} & a_x + \frac{b_x}{2} & 0 & - \omega \\ - a_y + \frac{b_y}{2} & a_y + \frac{b_y}{2} & \omega & 0 
                   \end{array}
                                                                                                                                          \right ) & \longrightarrow \xi_F := \xi(b_x,b_y, \nu, \omega, a_x, a_y ),  \label{formF}
    \end{align}
    where $F \in \mbox{Skew} ( \mathbb{M}^{1,3})$ is expressed in the orthonormal basis $\{ e_{\alpha} \}$ (specifically $F(e_{\nu}) = F^{\mu}{}_{\nu} e_{\mu}$ with $F^{\mu}{}_{\nu}$ being the  row $\mu$, column $\nu$ of the matrix above),
    $\xi(b_x,b_y, \nu,\omega,a_x, a_y)$ is given by \eqref{xiexp} and
    the coordinates of the plane $\Pi_u$ are renamed as 
    $\{ x := X^2, y := X^3\}$.
%     $\{ x := Y, y := Z\}$. 

    Given an (active) orthochronous Lorentz transformation $\Lambda(e_{\mu}) =
    \Lambda_{\mu}^{\nu} e_{\nu}$, we may consider the skew-symmetric endomorphism
    $F_{\Lambda} := \Lambda \circ F \circ \Lambda^{-1}$. The construction above guarantees that
      \begin{align*}
        \xi_{F_{\Lambda}} = \Xi^{\Lambda}_{\star} (\xi_F)
      \end{align*}
      where  $\Xi^{\Lambda}$ is the conformal diffeomorphism associated to
      the Lorentz transformation $\Lambda$. Let us restrict from now on to proper (i.e. orthochronous
        and orientation preserving) Lorentz transformations. Thus, $\Xi^{\Lambda}$ is an orientation preserving conformal diffeomorphism, and having fixed the coordinate system
        $\{x,y \} \in \mathbb{R}^2$, as well as $\z = \frac{1}{2} (x - iy)$,
        $\Xi^{\Lambda}$ is a M\"obius transformation. Thus there exists
        a pair $\pm \A \in SL(2,\mathbb{C})$ such that $\chi^{\pm \A(\Lambda)} = \Xi^{\Lambda}$. We are interested in determining the explicit form of $\A(\Lambda)$ (actually of
        its inverse map $\Lambda (\A)$). Having also fixed a future directed orthonormal basis $\{ e_{\alpha} \}$, we may represent a proper Lorentz transformation as an element of $SO^{\uparrow} (1,3)$ (the connected component of the identity of $SO(1,3)$). The aim is, thus,  to determine the map $\M: SL(2,\mathbb{C}) \rightarrow  SO^{\uparrow} (1,3)$  satisfying $\Xi^{\M(\A)} = \chi^{\A}$. Of course, this maps depends on the choices we have made concerning the unit spacelike direction $u$  and plane $\Pi_u$ to perform the stereographic projection.

        As discussed at length in many references, (see e.g. \cite{PenroseRindVol1}, pp. 8-24), when the vector $u$ is chosen to be $e_z$, the plane
        is selected to be $ \{ T=1, X^3=0\}$ and the complex coordinate
       $z'$ in this plane is taken as $z' = X^1 + i X^2$,  the corresponding map  $\M^{\prime}$ is  (we parametrize $\A$ is in \eqref{matrixA})
\begin{align*}
  \M^{\prime} (\A) = \frac{1}{2} \left (
  \begin{array}{cccc}
    \al \alb+ \be  \beb+ \ga \gab+ \del \delb &
                                      \al \beb 
                                      + \be \alb + \ga \delb
                                      + \del \gab ) & 
                                                   i (\al \beb - \be \alb +\ga \delb
                                                   - \del \gab ) & 
\al \alb-\be \beb+\ga \gab-\del \delb \\
    \al \gab +\be \delb + \ga \alb 
    +\del \beb &
              \al \delb +\be \gab + \ga \beb 
              + \del \alb  &
                          i (\al \delb -\be \gab+ \ga \beb 
                          - \del \alb ) &
                                       \al \gab -\be \delb + \ga \alb 
                                       - \del \beb  \\
    i (-\al \gab -\be \delb + \ga \alb 
    + \del \beb ) &
                 i (-\al \delb -\be \gab + \ga \beb 
                 + \del \alb ) & 
                              \al \delb -\be \gab- \ga \beb 
                              + \del \alb  & 
                                          i (-\al \gab +\be \delb + \ga \alb 
                                          - \del \beb ) \\
 \al \alb+\be \beb-\ga \gab-\del \delb & 
                               \al \beb + \be \alb -\ga \delb
                               - \del \gab  & 
                                           i (\al \beb - \be \alb -\ga \delb
                                           + \del \gab ) & 
\al \alb-\be \beb-\ga \gab+\del \delb
  \end{array} \right )
                                       \end{align*}
                                       We may take advantage of this fact to determine our $\M(\A)$. To do that we simply need to relate the action of the M\"obius group in the plane $\Pi_u := \{ X^1 = 1\}$ (in the coordinate $z$) with the
                                       corresponding
                                       action on the plane $\Pi'_u := \{X^3 =0\}$
                                       in the coordinate $z'$. At this point we can explain the reason why we have chosen $z = \frac{1}{2} (x-i y)$. The
                                       reason for the factor $2$ comes from the
                                       fact that the plane $\Pi_u$ lies at distance $2$ from the point of stereographic projection, while the plane
                                       $\Pi'_u$ lies at distance $1$ of its corresponding stereographic point. The sign is introduced because the basis $\{-e_1, e_2, e_3\}$ (with respect to which the point $u$ and the coordinates $\{x,y\}$ are
                                       defined) has opposite orientation than the basis $\{ e_3, e_1, e_2\}$ with respect to which the point $u'$
                                       and the coordinates $\{X^1, X^2\}$
                                       are built. By introducing a minus sign
                                       in $z$ we make sure that the transformation $\psi$ of $\mathbb{S}^2$ defined by $\{ z(p) = z'(\psi(p)) \}$ is orientation
                                       preserving (where $z(p)$ and
                                       $z'(p)$ stand for the two respective
                                       stereographic projections of
                                       $\mathbb{S}^2$ onto $\mathbb{C}^2
                                       \cup \{ \infty \}$). Now, a straightforward
                                       computation shows that an orientation
                                       preserving conformal diffeomorphism $\chi: \mathbb{S}^2 \rightarrow \mathbb{S}^2$ which in the plane $\Pi_u$ takes the
                                       form
                                       \begin{align*}
                                         \z( \chi(p)) =
                                         \frac{\alpha \z(p) + \beta}{\gamma \z(p) + \delta}, \quad \quad \alpha \delta - \beta \gamma = 1, \qquad  p\in \mathbb{S}^2
                                       \end{align*}
                                       has the following form in the
                                       $\Pi_u'$ plane
                                       \begin{align*}
                                         z'(\chi(p)) = 
\frac{\alpha' z'(p) + \beta'}{\gamma' z'(p) + \delta'}
\end{align*}
where
\begin{align*}
\left ( \begin{array}{cc}
\alpha' & \beta' \\
\gamma' & \delta'   
      \end{array}   
                  \right ) = U^{-1}  
\left (                      
\begin{array}{cc}
\alpha & \beta \\
\gamma & \delta   
      \end{array} 
\right )
         U, \quad \quad
         U: = \frac{1}{2}\left ( \begin{array}{cc}
                                  1 - i& - 1 +i \\
                                  1+ i & 1+ i
                                \end{array}
                                         \right ) .
\end{align*} 
Since the map $\M'$ is a morphisms of groups, it follows that the Lorentz transformation $\M(\A)$ is given by
\begin{align*}
  \M(\A) =\M'(\A') = \M'(U)^{-1} \M'(\A) \M'(U)
\end{align*}
The $SO^{\uparrow}(1,3)$ Lorentz matrix $\M'(U)$ is the rotation
\begin{align}
  \M'(U)  = \left ( \begin{array}{cccc}
                      1 & 0 & 0 & 0 \\
                      0 & 0 & 1 & 0 \\
                      0 & 0 & 0 & -1 \\
                      0 & -1 & 0 & 0 
                    \end{array}
                                  \right )
\end{align}
and we conclude that the Lorentz transformation $\M(\A)$
takes the explicit form 
\begin{align}
\M(\A) = 
 \frac{1}{2}  \left (
  \begin{array}{cccc}
 \al \alb + \be \beb + \ga \gab + \del \delb & - \al \alb + \be \beb - \ga \gab + \del \delb &
     \al \beb + \be \alb  + \ga \delb + \del \gab  & i ( - \al \beb + \be \alb \ - \ga \delb + \del \gab)\\
     - \al \alb - \be \beb + \ga \gab + \del \delb & \al \alb - \be \beb - \ga \gab + \del \delb &
      - \al \beb - \be \alb  + \ga \delb + \del \gab  & i (\al \beb - \be \alb  - \ga \delb + \del \gab )\\
    \al \gab + \be \delb + \ga \alb  + \del \beb  & - \al \gab + \be \delb - \ga \alb + \del \beb & \al \delb + \be \gab + \ga \beb + \del \alb 
                                                   & i ( - \al \delb 
                                                     + \be \gab - \ga \beb + \del \alb)\\
    i (\al \gab + \be \delb - \ga \alb  - \del \beb ) & i ( - \al \gab + \be \delb + \ga \alb - \del \beb ) &     i (\al \delb 
                                                                                                              + \be \gab - \ga \beb - \del \alb ) & \al \delb                                                                                                                                                     - \be \gab - \ga \beb  + \del \alb\\
  \end{array}
  \right ) \label{MA}
\end{align}
(to avoid ambiguities, recall  that the Lorentz transformation defined by this matrix is
$\Lambda(e_{\mu} ) = \Lambda^{\nu}_{\mu} e_{\nu}$ with $\Lambda^{\nu}_{\mu}$
the row $\nu$ and column $\mu$).

\section{Canonical form of the GCKV}\label{seccanonGCKV}

We start with a definition motivated by the canonical form of skew-symmetric
endomorphisms discussed in Section \ref{CanonSkew}.
\begin{definition}\label{defcanongkvf}
  Let $\mathbb{E}^2$ be Euclidean space and $\{x,y\}$ a

  Cartesian coordinate system. A GCKV $\xi$ is called {\bf canonical} with
  respect to $\{x,y\}$ if it has the form
  \begin{align*}
    \xi = (\mu_0 + z^2 ) \partial_z + \left ( \overline{\mu_0} + \overline{z}^2
    \right ) \partial_{\overline{z}}, \quad \quad   \z:= \frac{1}{2} (x-iy),
    \quad \mu_0 \in \mathbb{C}.
  \end{align*}
\end{definition}
Equivalently, a GCKV is canonical with respect to $\{x,y\}$ whenever its corresponding
form \eqref{muform} has $\mu_1 =0$ and $\mu_2=2$. We next characterize the  class of
M\"obius transformations $\chi^\A$ which send a
given GCKV into its canonical form.

\begin{proposition}
\label{transf}
  Let $\{x,y\}$ be a Cartesian coordinate system in $\mathbb{E}^2$.
  Let $\xi$ be a non-trivial GCKV and define the complex constants $\{\mu_0, \mu_1, \mu_2\}$ 
  such that $\xi = \xi_{\{\mu\}}$ when expressed in the complex coordinate
  $\z = (x -i y)/2$ and its complex conjugate. Then $\chi^\A \in \Mo$ has the property that $\chi^\A_{\star} (\xi)$ is written in canonical form with respect
  to $\{ x,y\}$ if and only if
  \begin{align}
    \A = \left ( \begin{array}{cc}
                   \frac{1}{2} \left ( \delta \mu_2 - \gamma \mu_1 \right ) &
                                                                              \frac{1}{2} \delta \mu_1 - \gamma \mu_0 \\
                   \gamma & \delta
                 \end{array}
                            \right ), \quad
                            \frac{1}{2} \delta^2 \mu_2 - \gamma \delta \mu_1 + \gamma^2 \mu_0 =1. \label{formA}
  \end{align}
  Moreover, for any such $\A$, it holds
  \begin{align*}
    \chi^{\A}_{\star} (\xi) =
    \left ( \frac{1}{4} \left ( \aa_{\{\mu\}} - i \bb_{\{\mu\}} \right ) + z^2 \right ) \partial_{z}
    +
    \left ( \frac{1}{4} \left ( \aa_{\{\mu\}} + i \bb_{\{\mu\}} \right ) + \overline{z}^2 \right ) \partial_{\overline{z}}.
  \end{align*}
  
  \end{proposition}

\begin{proof}
From \eqref{transmu} and the fact that the canonical form
  has $\mu'_1=0$ and $\mu'_2=2$, we need to find the most general $\alpha, \beta, \gamma, \delta$ subject to
  $\alpha\delta- \beta\gamma=1$ such that
  \begin{align}
    - 2 \alpha \gamma \mu_0 + \left ( \alpha \delta + \beta \gamma
    \right ) \mu_1 - \beta \delta \mu_2 & =0, \label{firstEq}\\
    2 \gamma^2 \mu_0 - 2 \gamma \delta \mu_1 + \delta^2 \mu_2 &= 2.
                                                                \label{secondEq}
  \end{align} 
  The first can be written, using the determinant condition
  $\alpha\delta- \beta\gamma=1$,
  as
  $- 2 \alpha \gamma \mu_0 + (1+ 2 \beta \gamma) \mu_1 - \beta \delta \mu_2=0$. Multiplying by $\delta$ yields
  \begin{align}
    0 & = -2 \alpha \delta \gamma \mu_0 + \delta \mu_1 + \beta
    \left ( 2 \gamma \delta \mu_1 - \delta^2 \mu_2 \right )
    = -2 \alpha\delta \gamma \mu_0 + \delta \mu_1 + \beta \left (
        2 \gamma^2 \mu_0 -2 \right ) \nonumber \\
    & =
    - 2 \gamma \mu_0 + \delta \mu_1 - 2 \beta \quad \quad 
  \Longrightarrow \quad \quad \beta = \frac{1}{2} \delta \mu_1 - \gamma \mu_0,
\label{beta}
  \end{align}
  where in the second equality we used  \eqref{secondEq}
  and in the third one we inserted the determinant condition. To determine
  $\alpha$ we compute
  \begin{align*}
    & \alpha \delta = 1 + \beta \gamma = 1 + \frac{1}{2} \gamma \delta \mu_1  - \gamma^2 \mu_0  =
    \frac{1}{2} \delta \left ( \delta \mu_2 - \gamma \mu_1 \right ) \\
   \quad \quad \Longrightarrow \quad \quad
    & \delta \left ( \alpha + \frac{1}{2} \gamma \mu_1
    - \frac{1}{2} \delta \mu_2 \right )=0,
  \end{align*}
  where in the third equality we used \eqref{secondEq} to replace
  $\gamma^2 \mu_0$. If $\delta \neq 0$ we conclude that
  $\alpha = (1/2) (\gamma \mu_1 - \delta \mu_2)$, and the form
    of $\A$ is necessarily as given in \eqref{formA}. If, on the other hand,
$\delta =0$, then the determinant condition forces
  $\gamma \neq 0$. Thus, equation \eqref{firstEq} gives
  $- 2 \alpha \mu_0' + \beta \mu_1=0$, which after using \eqref{beta} implies $\alpha = - (1/2) \gamma \mu_1$, so \eqref{formA} also follows. This proves the
  ``only if'' part of the statement. For the ``if'' part one simple checks that
  $\beta$ and $\alpha$ obtained above indeed satisfy \eqref{firstEq}-\eqref{secondEq}, as soon as $\gamma, \delta$ satisfy the determinant condition
  given in \eqref{formA}. 

  The second part of the Proposition is immediate form the fact that
  $2 \mu_0 \mu_2  - \mu_1^2$ is invariant under  \eqref{transmu}. Thus,
  $\chi^{\A}_{\star} (\xi)$ has $\mu_0'$ satisfying
  \begin{align}
    4 \mu_0' = 2 \mu_0' \mu_2' - \mu_1'{}^2 
    = 2 \mu_0 \mu_2 - \mu_1^2
    = \aa_{\{\mu\}} - i \bb_{\{\mu\}}.
    \label{aabb}
\end{align}  
\end{proof}

\begin{corollary}
  \label{inv1}
  The subgroup of $SL(2,\mathbb{C})$  that leaves invariant
  a GCKV field in canonical form with parameter $\mu_0$ is given by
  \begin{align}
    \A_{\mu_0} = \left \{  \left ( \begin{array}{cc}
                   \delta  & -  \gamma \mu_0 \\
                   \gamma & \delta
                 \end{array}
                            \right ), \quad
                            \delta^2 + \mu_0 \gamma^2 = 1 \right \}
                                                        . \label{invA}
  \end{align}

\end{corollary}

{\it Proof.} Insert $\mu_1=0$ and $\mu_2=2$ into \eqref{formA}.

\begin{corollary}
  \label{coset}
  Given any GCKV $\xi$ as in Proposition \ref{transf}, the set of elements
  $\A \in SL(2,\mathbb{C})$ such that $\chi^{\A}_{\star}(\xi)$
  takes the canonical form is
  \begin{align*}
    \A_{\frac{1}{4} (\aa_{\{\mu\}} - i \bb_{\{\mu\}} )}
    \cdot \A_0
  \end{align*}
  where $\A_{0}$ is any element of $SL(2,\mathbb{C})$ satisfying
  \eqref{formA}.
\end{corollary}
{\it Proof.} Fix $\A_0$ satisfying \eqref{formA}. Any other element
$\A_1$ will satisfy \eqref{formA} if and only if
$\A_1 \A_0^{-1}$ leaves invariant the column vector $(\mu_0',0,2)$,
$4 \mu_0' := \aa_{\{\mu\}} - \bb_{\{\mu\}}$, i.e. if and only if $\A_1 \cdot \A_0
\in \A_{\mu_0'}$.

\vspace{5mm}

In the next corollary, we denote the entries of a matrix  $(U)$ by $U^{\mu}{}_{\nu}$, where the upper index $\mu$ refers to row and the lower index $\nu$ refers to column.
%%%%%%%%%%%%%%%%%%%%%%%%%%%%%%%5
%%% The following corollary still needs to be contrasted ith the file theorem2.red
%%%%%%%%%%%%%%%%%%%%%%%%%%%%%%%%%%%
\begin{corollary}\label{canonbasescoord}
  Let $F$ be a non-zero skew-symmetric endomorphism in $\mathbb{M}^{1,3}$ and let the matrix $(F)$ be defined by 
  $F(e_{\mu}) = F^{\nu}{}_{\mu} e_{\nu}$ where $\{ e_{\mu}\}$ is an orthonormal basis. Define $\{ b_x, b_y,\nu, \omega, a_x, a_y\}$
  so that  $(F)$ reads as in \eqref{formF}. Define
  $\mu_0, \mu_1, \mu_2$ by means of \eqref{mudata}  and let $\Lambda := \M(\A)$, where $\A$ is any of the matrices
  defined in Proposition \ref{transf}.  Then, in the basis $e'_{\nu} := \Lambda^{\mu}{}_{\nu} e_{\mu}$, the endomorphism $F$ takes the canonical form
  \eqref{canonFdim4}
  with $\aa - i \bb = 2 \mu_0 \mu_2 - \mu_1^2$. 
\end{corollary}

   In Proposition \ref{propcanonF4} we showed the existence of the canonical form of $F\in \skwend{\mathbb{M}^{1,3}}$, and this motivated the Definition \ref{defcanongkvf} of canonical form of GKVFs. However, it is only in Corollary \ref{canonbasescoord} that we have been able to (easily) find the explicit change of basis that takes $F$ to its canonical form. This is possible because we are dealing with low dimensions and the GCKVFs take a very simple expression in complex coordinates of the Riemann sphere,
     but this is a much more difficult problem in higher dimensions.

   We can however easily derive the three-dimensional case
       as a simple consequence. For that we consider,  as usual, the extension $\widehat F \in \skwend{\mathbb{M}^{1,3}}$ of $ F \in \skwend{\mathbb{M}^{1,2}}$ described before Corollary \ref{propcanonF3}.
     In the basis $\lrbrace{e_0,e_1,e_2,e_3:=E_3}$, $\widehat F$ has $a_y = b_y = \omega = 0$, so the quantities $\mu_0, \mu_1, \mu_2$ defined in \eqref{mudata} are real. In order to apply Corollary \ref{canonbasescoord}
       to find the change of
       orthonormal basis $\{e_0,e_1,e_2\}$ that brings $F$ into its canonical form we simply need to impose that $e_3' = e_3$,
     which amounts to ${\Lambda^0}_3 = {\Lambda^1}_3 = {\Lambda^2}_3 = 0$ and ${\Lambda^0}_3 = 1$.   It is easy to show (recall that
       $\alpha,\beta$ are expressed in terms of $\gamma, \delta$ in the matrix
$\A$ of Corollary \ref{canonbasescoord}) that
     the general solution to the first three equations is $\gamma \bar \delta = \bar \gamma \delta$. The condition $\Lambda^0{}_3=1$ is then
     \begin{align*}
       \frac{1}{2} \delta \bar \delta \mu_2 -  \gamma \bar\delta \mu_1 + \gamma \bar \gamma \mu_0 = 1.
     \end{align*}
     Multiplying by $\delta$ and using the determinant condition
     in \eqref{formA} implies $\delta = \bar{\delta}$, while multiplying by $\gamma$ gives $\gamma = \bar\gamma$, and then
     $\Lambda^0{}_3 =1$ is just identical to the determinant condition  so no more consequences can be extracted. Thus all parameters $\alpha,\beta, \gamma, \delta$ are real.
   Summarizing:
   \begin{corollary} 
  Let $F$ be a non-zero skew-symmetric endomorphism of $\mathbb{M}^{1,2}$ and the matrix $(F)$ be defined by 
  $F(e_{i}) = F^{j}{}_{i} e_{j}$ where $\{ e_{i}\}_{i=0,1,2}$ is an orthonormal basis. Define $\mu_0 := ({F^1}_3 - {F^2}_3)/2,~\mu_1 := -{F^1}_2,~\mu_2 := -({F^1}_3 + {F^2}_3)$. For
    any pair of real numbers $\gamma, \delta$ satisfying $\delta^2 \mu_2 - 2 \gamma \delta \mu_1 + 2 \gamma^2 \mu_0 = 2$, let $\alpha := (\delta \mu_2 - \gamma \mu_1)/2$ and $\beta := \delta \mu_1/2 - \gamma \mu_0$.
    %for $\gamma, \delta \in \mathbb{R}$ such that $\alpha \delta - \beta \gamma =1$. 
  Then, in the basis $e'_{i} := \Lambda^{j}{}_{i} e_{j}$, with 
  \begin{equation}
   \Lambda := \left(
\begin{array}{ccc}
 \frac{1}{2} \left(\alpha ^2+\beta ^2+\gamma ^2+\delta ^2\right) & \frac{1}{2} \left(-\alpha ^2+\beta ^2-\gamma ^2+\delta ^2\right) & \alpha  \beta +\gamma  \delta  \\
 \frac{1}{2} \left(-\alpha ^2-\beta ^2+\gamma ^2+\delta ^2\right) & \frac{1}{2} \left(\alpha ^2-\beta ^2-\gamma ^2+\delta ^2\right) & -\alpha\beta + \gamma  \delta  \\
 \alpha  \gamma +\beta  \delta  & -\alpha\gamma + \beta  \delta  & \alpha \delta + \beta  \gamma  \\
\end{array}
\right),
\end{equation}
  the endomorphism $F$ takes the canonical form
  \eqref{canonFdim3}
  with $\aa = 2 \mu_0 \mu_2 - \mu_1^2$. 
\end{corollary}
%    \begin{corollary}
%     Let $F$ be a non-zero skew-symmetric endomorphism of $\mathbb{M}^{1,2}$ and consider the isometric identification $\mathbb{M}^{1,3} \cong \mathbb{M}^{1,2} \oplus \mathbb{E}_1$. 
%     Define $\hat F \in \skwend{\mathbb{M}^{1,3}}$ by $\restr{\hat F}{\mathbb{M}^{1,2}} = F$ and $\restr{\hat F}{\mathbb{E}_1} = 0$ and let $e'_{\nu} := \Lambda^{\mu}{}_{\nu} e_{\mu}$ the canonical basis of $\hat F$given in Corollary \ref{canonbasescoord} with the additional constraint $\gamma \bar \delta =  (\mu_1^2 - \mu_0 \mu_2)/4$. Then in the subasis $\lrbrace{e'_1,e'_2,e'_3}$, $F$ takes the canonical form  \eqref{canonFdim3}. Moreover the basis is given by $e'_{j} := \Lambda^{i}{}_{j} e_{i}$, for $i,j = 0,1,2$, where $\Lambda^{i}{}_{j}$ defines a Lorentz transformation of $\mathbb{M}^{1,2}$.
%    \end{corollary}

\section{Adapted coordinates to a GKCV}\label{secadaptedcoords}

So far we have explored the action of the M\"obius group on a GCKV and have found that for any such vector, there exists a set of transformations that brings
it into a canonical form. The perspective so far has been active. We now change
the point of view and exploit the previous results to find coordinate systems
in (appropriate subsets of) $\mathbb{E}^2$ that rectify a given (and fixed) GKCV $\xi$.

Consider $\mathbb{E}^2$ and fix a non-trivial GCKV field $\xi$.
Let us select a Cartesian coordinate system $\{ x,y\}$
and define, as before $z = (1/2) (x-iy)$ and $\overline{z} = (1/2) ( x+ i y)$.

When expressed in the $\{ z, \overline{z}\}$ coordinate system $\xi$ will be $\xi = \xi_{\{\mu\}}$
for some triple of complex numbers $\{ \mu \} = \{ \mu_0, \mu_1, \mu_2\}$. We now view the M\"obius transformation as a change of coordinates. Specifically, given
$\alpha,\beta,\gamma, \delta$ complex constants satisfying
$\alpha\delta - \beta \gamma =1$, the quantity
\begin{align}\label{omegaz}
  \omega =  \frac{\alpha z + \beta}{\gamma z + \delta}
\end{align}
and its complex conjugate $\overline{\omega}$ define a coordinate system on $\mathbb{R}^2 \setminus \{ \gamma z+ \delta =0 \}$. The inverse of this coordinate transformation is, obviously,
  \begin{align}
    z = \frac{\delta \omega - \beta}{- \gamma \omega + \alpha}.
    \label{z-omega}
  \end{align}
  It is well-known that transformations of a manifold can be dually seen as coordinate changes in suitable restricted coordinate patches. We will refer to
  \eqref{z-omega} as a M\"obius coordinate change.
  With this point of view, we may express $\xi$ in the coordinate system $\{ \omega,\overline{\omega}\}$ and the duality above implies that $\xi$ takes the form 
  \begin{align*}
    \xi =
    \left ( \mu_0' + \mu_1' \omega + \frac{1}{2} \mu_2' \omega^2
    \right ) \partial_{\omega} +
    \left (  \overline{\mu_0'} + \overline{\mu_1'} \overline{\omega} + \frac{1}{2} \overline{\mu_2'} \overline{\omega}^2
    \right ) \partial_{\overline{\omega}} 
      \end{align*}
      with $\{\mu_0', \mu_1', \mu_2'\}$ given by \eqref{transmu} (this can also be checked by direct computation).

      We may now take $\{\alpha,\beta,
      \gamma, \delta\}$ so that corresponding matrix $\A$ satisfies \eqref{formA}. It follows that $\xi$ takes the canonical form
      \begin{align}
        \xi :=
        \left ( \frac{1}{4} \left ( \aa_{\{\mu\}} - i \bb_{\{\mu\}} \right )+ \omega^2
        \right ) \partial_{\omega}
        + \left ( \frac{1}{4} \left ( \aa_{\{\mu\}} + i \bb_{\{\mu\}} \right ) + \overline{\omega}^2
        \right ) \partial_{\overline{\omega}}. \label{canonical}
      \end{align}
      By  Lemma \ref{orto}, the vector $\xi^{\perp}$ defined by
      $\xi^{\perp} := \xi_{\{  i \, \mu \}}$  is a GCKV
      orthogonal to $\xi$ everywhere, with the same pointwise norm as $\xi$
      and satisfying $[\xi, \xi^{\perp} ] =0$. In particular
      $\xi$ and $\xi^{\perp}$ are linearly independent except at points where both vanish identically. As a consequence, it makes sense to tackle the problem of finding coordinates that rectify $\xi$ by trying to determine a coordinate system
      $\{ v_1, v_2 \}$ (on a suitable subset of $\mathbb{R}^2$) such that
      \begin{align*}
        \xi = \partial_{v_1}, \quad \quad \xi^{\perp} = \partial_{v_2}.
      \end{align*}
      Assume that we have already transformed into the coordinates $\{ \omega,
      \overline{\omega}\}$ where $\xi$ (and also
      $\xi^{\perp}$) take their canonical forms
      \begin{align}
        \xi = \left ( \frac{1}{4} Q e^{-2 i \arg} + \omega^2 \right )
        \partial_{\omega} + \mbox{c.c}, \quad \quad
        \xi^{\perp} = \left ( \frac{i }{4} Q e^{ - 2 i \arg} + i \omega^2 \right ) \partial_{\omega} + \mbox{c.c} \label{canform}
      \end{align}
      where we have defined the real constants $Q \geq 0$ and $\arg \in [0,\pi)$ by
      \begin{align}
        \aa_{\{\mu\}} - i \bb_{\{\mu\}} =Q e^{-2 i \arg} \label{defQsigma}
      \end{align}
      and  where $\mbox{c.c.}$ stands for complex conjugate of the previous term. We are seeking a coordinate system $\{ \zeta, \overline{\zeta}\}$ defined by
              \begin{align*}
                \zeta := \frac{1}{2} \left ( v_1 + i v_2 \right )
              \end{align*}
              such that
              \begin{align*}
              \xi - i \xi^{\perp} = \partial_{\zeta} 
\end{align*}
(this is because $\partial_\zeta = \partial_{v_1} - i \partial_{v_2}$). Since
$\xi - i \xi^{\perp} = 2 \left ( \frac{1}{4} Q e^{-2 i \arg} + \omega^2 \right )
\partial_{\omega}$ the coordinate change musty satisfy the ODE
\begin{align*}
  \frac{d \zeta}{d \omega} = \frac{1}{2 \omega^2 + \frac{Q}{2} e^{-2 i \arg} }.
  \end{align*}
  This equation can be integrated immediately. The result is
  \begin{align}
      \zeta(\omega) =  \zeta_0 + \frac{-i e^{i \arg}}{2 \sqrt{Q}}
      \ln \left (
      \frac{\omega - i \frac{\sqrt{Q}}{2} e^{- i \arg}}
    {\omega + i \frac{\sqrt{Q}}{2} e^{- i \arg}} \right )
    %& 
      \quad \quad \Longleftrightarrow \quad \quad
     % \nonumber \\
     %&
       \omega(\zeta;\zeta_0) = \frac{i \sqrt{Q} e^{- i \arg}}{2} \frac{1 + e^{2 i \sqrt{Q} e^{- i \arg} (\zeta- \zeta_0)}}{1-  e^{2 i \sqrt{Q} e^{- i \arg} (\zeta-\zeta_0)}},
                                                              \label{omega-zeta}
   %   = \frac{-i e^{-i \arg}}{2 \sqrt{Q}} \ln (\s)
  \end{align}
  where $\zeta_0$ is an arbitrary complex constant.
  These expressions include the case $Q=0$ as a limit.  Explicitly
  \begin{align}
    \zeta - \zeta_0 = - \frac{1}{2 \omega} \quad \quad \Longleftrightarrow \quad \quad \omega = - \frac{1}{2 (\zeta- \zeta_0)}.
    \label{Q=0}
    \end{align}
Since the logarithm is a multivalued complex function, one needs to be
   careful concerning the domain and range of this coordinate change. 
In the $\{\omega, \overline{\omega} \}$ plane, the vector field $\xi$  vanishes at the two points (cf.
\eqref{canform}) $\omega = \pm i \frac{\sqrt{Q}}{2} e^{-i \arg}$ (which degenerate to the point at the origin when $Q=0$). It is clear that neither of these
points will be covered by the $\{ \zeta, \overline{\zeta}\}$ coordinate system. The case $Q=0$ is very simple because, from \eqref{Q=0}, it is clear
that the $\{ \zeta, \zeta \}$ coordinate system covers the whole $\{ \omega, \overline{\omega} \}$ plane except the origin.  Since the point at infinity
in the $\omega$-plane is sent to the point $\zeta_0$ in the $\zeta$-plane we conclude that the $\{ \zeta, \overline{\zeta}\}$ coordinate covers the whole Riemann sphere except the single point where $\xi$ vanishes.     

When $Q \neq 0$, the situation  is more interesting. The reason in the multivaluedness of the logarithm. This suggests that the coordinate change may in
fact define a larger manifold that covers the original one. In order to discuss this, let is introduce 
the auxiliary function
  \begin{align*}
    \s := \frac{\omega - i \frac{\sqrt{Q}}{2} e^{- i \arg}}
    {\omega + i \frac{\sqrt{Q}}{2} e^{- i \arg}}.
    \end{align*}
    This is a M\"obius transformation, so it maps diffeomorphically $\mathbb{C} \cup \{ \infty \}$ onto itself. The two zeroes of $\xi$ are mapped
    respectively to the origin and infinity in the $\s$ variable.  Since
    \eqref{omega-zeta} can be written as
    $    \zeta - \zeta_0= - i e^{i \arg} \ln ( \s) / (2 \sqrt{Q})$ and $\ln (\s) = \ln |\s| + i (\mbox{arg}(\s) + 2 \pi m), m \in \mathbb{N}$,  a single value of $\s$
    may be mapped to an infinite number of points depending on the branch on the branch of logarithm one takes.
    One may decide to restrict the $\{ \zeta, \overline{\zeta}\}$-domain to be the band $B:= \{ \zeta \in \mathbb{C} : \mbox{Im} (2 i \sqrt{Q} e^{-i \arg} (\zeta - \zeta_0)) \in (0,2\pi) \}$ and then the coordinate change $\zeta (\s)$ defines a diffeomorphism between $\mathbb{C} \setminus \{ \s = (r,0), r \geq 0  \}$
    into $B$.
    Let $\partial_1 B$ be the connected component of $\partial B$ defined by
    $\mbox{Im} (2 i \sqrt{Q} e^{-i \arg} (\zeta - \zeta_0)) =0$ and $\partial_2 B$ the other component 
    $\partial_2 B := \{ \mbox{Im} (2 i \sqrt{Q} e^{-i \arg  } (\zeta- \zeta_0)) = 2 \pi\}$, then
      the semi-line $\{ \s  = r\}$, with $r$ real and positive  and $\mathrm{arg}(\mathfrak{z})\in \lrbrace{0, 2 \pi}$, is mapped to the respective points
      $\zeta_1(r) = - i e^{i \arg} \ln(r) /(2\sqrt{Q}) \in \partial_1 B$ and
       $\zeta_2(r) = - i e^{ i \arg} \ln(r) /(2\sqrt{Q}) + \pi e^{i \arg}/\sqrt{Q} \in \partial_2 B$. This shows that
      these two boundaries are to identified  by means of the translation defined by the shift 
      \begin{align}
        \zeta_t := \pi e^{i \arg} /\sqrt{Q}. \label{translation}
      \end{align}
      The topology of the resulting
    manifold is $\mathbb{R} \times \mathbb{S}^1$. This is in agreement with the fact that $\xi$ vanishes at precisely two points of the Riemann sphere, and the complement of two points on a sphere is indeed a cylinder. The alternative is to let $\zeta$
    take values in  all $\mathbb{C}$ and consider the inverse map
    \begin{align*}
      \s (\zeta) := e^{2 i \sqrt{Q} e^{- i\arg} (\zeta - \zeta_0)}.
    \end{align*}
    It is clear that this defines an infinite covering of the $\s$-punctured complex plane $\mathbb{C} \setminus \{ 0\}$.  As described above,
   the fundamental domain of this covering is the (open) band $B$ limited by the lines (see figure \ref{band}, where we have set $\zeta_0=0$ for definiteness)
        \begin{align*}
          \zeta_1(s) & =  \zeta_0 + \frac{- i e^{i \arg}  s}{2 \sqrt{Q}}, \quad  \quad \quad \quad  s \in \mathbb{R}, \\
          \zeta_2(s) & = \zeta_0 + \frac{- i e^{ i \arg}  s}{2 \sqrt{Q}} + \zeta_t, \quad  \quad s \in \mathbb{R}.
        \end{align*}
  
        \begin{figure}   
          \begin{center}
          \label{band}
          \psfrag{B}{$B$}
          \psfrag{B1}{$\partial_1 B$}
          \psfrag{B2}{$\partial_2 B$}
\psfrag{z0}{$\frac{2 \pi}{\sqrt{Q}}$}          
  \psfrag{sig}{$\arg$}
  \psfrag{xi}{$\xi$}
  \psfrag{xip}{$\xi^{\perp}$}
  \psfrag{v1}{$v_1$}
  \psfrag{v2}{$v_2$}  
  \includegraphics[height=7cm]{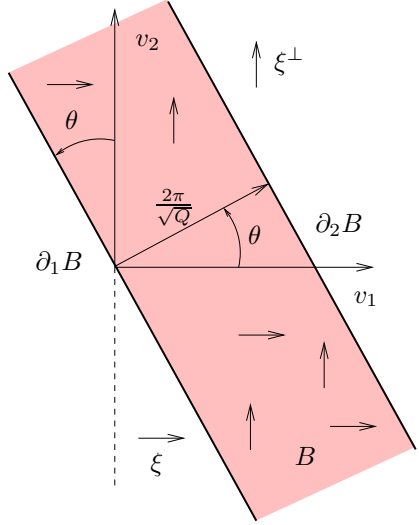}  
  \caption{Domain of the complex coordinate $\zeta = \frac{1}{2} (v_1 + i v_2)$ adapted to $\xi = \partial_{v_1}$ and $\xi^\perp = \partial_{v_2}$. The parameters $Q$ and $\arg$ determine the width and tilt of the band respectively. The factor two in the distance between the boundaries (compare \eqref{translation}) arises because $\zeta = \frac{1}{2} (v_1 + i v_2)$. }
\end{center}
\end{figure}

        \begin{figure}  
          \begin{center}
          \label{lines}
          \psfrag{B}{$B$}
          \psfrag{B1}{$\partial_1 B$}
          \psfrag{B2}{$\partial_2 B$}
\psfrag{z0}{$\frac{2\pi}{\sqrt{Q}}$}          
  \psfrag{sig}{$\arg$}
  \psfrag{xi}{$\xi$}
  \psfrag{xip}{$\xi^{\perp}$}
  \psfrag{v1}{$v_1$}
  \psfrag{v2}{$v_2$}
  \includegraphics[height=7cm]{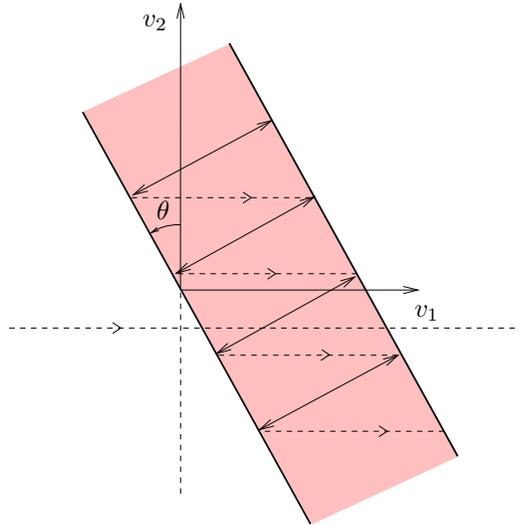}
  \caption{Integral lines of $\xi$ (dashed line). The points joint by arrows are identified by the translation defined by \eqref{translation}.   }
  \end{center}
\end{figure}    
         
        The $\zeta$-complex plane therefore corresponds to the complete unwrapping        of the cylinder, i.e. to its universal covering.
                In the $\{ \zeta, \overline{\zeta}\}$ coordinate system we have
\begin{align*}
  \xi = \frac{1}{2} \left ( \partial_\zeta + \partial_{\overline{\zeta}} \right ), \quad \quad
  \xi^{\perp} = \frac{i}{2} \left ( \partial_\zeta - \partial_{\overline{\zeta}} \right ),
\end{align*}
so $\xi$ points along the real axis and $\xi^{\perp}$ into the imaginary axis. The angle of the boundaries $\partial_1 B$ (and $\partial_2 B$) with the
real axis is $\frac{\pi}{2} + \arg$. For generic values of
$\arg$ it follows that the integral lines
of $\xi$ descend to the quotient $\overline{B}$ (with the boundaries
identified as above) as open lines that asymptote to the two points at
infinity along the band
(as in Figure \ref{lines}). Observe that these two asymptotic values correspond  to 
$\s = 0$ or $\s = \infty$, which  correspond  to the two zeros
of $\xi$. Thus, the integral lines of $\xi$ start asymptotically at one of its zeros and  approaches asymptotically  the other zero. Along the way, the integral lines circle each zero an infinite number of times (because the projection
to the lines parallel to the real axis descend to the quotient in such a way
that they intersect the boundaries of $B$ an infinite number of times). The only exception to this behaviour is when 
$\arg  = \frac{\pi}{2}$ or when $\arg  =0$ (recall that by construction $\arg \in [0,\pi)$). In the former case, the 
integral lines of $\xi$, never leave the fundamental domain. This means that the curves asymptote to the two zeros of $\xi$ and they never encircle them along the way. The case $\arg = 0$ corresponds to the situation when
the projection of the integral lines of $\xi$ define closed curves
on $\overline{B}$ with the boundaries identified . This is the situation when the integral curves of $\xi$ in the
original $\{ \omega, \overline{\omega}\}$ plane are topological circles (which
degenerate to points at the zeroes of $\xi$).

It is interesting to see how the limit $Q=0$ is recovered in this setting.
The translation vector that identifies points in the boundary $\partial_1 B$
with points in the boundary $\partial_2 B$ diverges as $Q \rightarrow 0$. Thus, the band $B$ becomes larger
and larger until it covers the whole $\zeta$-plane in the limit. On other words, the $\zeta$-coordinate is no longer a covering of the original
$\omega$-coordinate. In the limit, $\xi$ vanishes at only one point
in the $\omega$-plane (the origin) which is sent to infinity in the
$\zeta$-coordinates. It is by the process of the band $B$ becoming wider and wider
that the limits at infinity along the band, which correspond to two points for any
non-zero value of $Q$, merge into a single point when $Q=0$.  The process also explains in which sense the parameter $\arg$, which measures the inclination
of the band $B$ becomes irrelevant in the limit $Q =0$, in agreement with the fact that \eqref{defQsigma} lets $\arg$ take any value when $\aa_{\{\mu\}} - i \bb_{\{\mu\}}$ (and hence also $Q$) vanishes.

In all the expressions above we have maintained the additive integration
constant $\zeta_0$, instead of setting it to zero as the simplest choice.
The reason is that $\zeta_0$ can be directly connected with the freedom one has in performing the coordinate change \eqref{z-omega} that brings $\xi$ into its canonical form. To understand this we simply note that, from \eqref{omega-zeta} one can check that the following identity holds
%%%%%%%%%%%%%%%%%%%%%%%%%%%%%
%%% see file adapted3.red%%%%
%%%%%%%%%%%%%%%%%%%%%%%%%%%%%
\begin{align*}
  \omega(\zeta; \zeta_0) = \frac{
  \cos  \left ( \sqrt{Q} e^{-i\arg} \zeta_0 \right )
  \omega(\zeta;0)  - \frac{\sqrt{Q}}{2} e^{-i \arg}
  \sin \left ( \sqrt{Q} e^{-i\arg} \zeta_0 \right )}{\frac{2}{\sqrt{Q}}  e^{i \arg}
  \sin \left ( \sqrt{Q} e^{-i\arg} \zeta_0 \right ) \omega(\zeta;0)+
  \cos  \left ( \sqrt{Q} e^{-i\arg} \zeta_0 \right )}.
  \end{align*}
  Thus, the relation between $\omega(\zeta;0)$ and
  $\omega(\zeta;\zeta_0)$ is a M\"obius transformation defined by the matrix
  \begin{align*}
    \left (     \begin{array}{cc}
                  \cos \left ( \sqrt{Q} e^{-i\arg} \zeta_0 \right ) & 
  - \frac{\sqrt{Q}}{2} e^{-i \arg}
                                                                        \sin \left ( \sqrt{Q} e^{-i\arg} \zeta_0 \right ) \\
                    \frac{2}{\sqrt{Q}}  e^{i \arg}
                                                                        \sin \left ( \sqrt{Q} e^{-i\arg} \zeta_0 \right )                              &
                                                                        \cos \left ( \sqrt{Q} e^{-i\arg} \zeta_0 \right )
                \end{array}
                                                                        \right ).
  \end{align*}
  It is immediate to check that, letting $\zeta_0$ take any
  value, one runs along the full subgroup $\A_{\frac{1}{4} Q  e^{-2 i \arg}}$ defined in Corollary \ref{inv1}. Thus, by Corollary \ref{coset}, the freedom in performing the coordinate change  \eqref{z-omega} that transforms
    $\xi$ into its canonical form can be absorbed into the additive constant
    $\zeta_0$, and vice-versa. Having understood this, we will set $\zeta_0 =0$
     from now on.

So far we have considered $\xi$ without referring  to any specific metric. We
now endow $\mathbb{R}^2$ coordinated by $\{x,y\}$ (or $\{ z, \overline{z}\}$) with the following class of metrics. Let $u :=\{u_0, u_1, u_2, u_3\} \in \mathbb{R}^4$, $u \neq 0$, and define
\begin{align}
  g_u &:= \frac{1}{\Omega_u^2} \left (dx^2 + dy^2  \right ) = \frac{1}{\Omega_u^2} 4 dz d \overline{z}, \label{metgu} \\
  \Omega_u &:= u_0 + u_1 + u_2 x + u_3 y + \frac{1}{4} (u_0 - u_1) (x^2+ y^2)
   =  u_0 (1+ z \overline{z}) +  u_1 (1 - z \overline{z} ) + u_2 (z + \overline{z} ) + u_3 i ( z - \overline{z}) . \nonumber %\label{gudef} 
\end{align}
The Gauss curvature of $g_u$ is $\kappa_u:= u_0^2 -u_1^2 - u_2^2 - u_3^2$.
Since $g_{-u} = g_{u}$, there is a sign freedom in $u$ that we must keep in mind.
When $\kappa_u \geq 0$, then it must be that $u_0 \neq 0$ and the sign freedom may be fixed by the requirement $u_0 > 0$. However, this is no longer possible when $\kappa_u <0$.

Observe that  $g_{\{ u_0=\frac{1}{2}, u_1=\frac{1}{2}, u_2=0, u_3= 0\} } = g_E := 4 dz d \overline{z}$.
%We also particularize one of the metrics
%of constant curvature equal to one, namely the metric $g_{\{u_0=1. u_1=0.u_2=0.u_3=0\}} := g_{\mathbb{S}^2}$ which we call the standard spherical metric.
Under a
M\"obius coordinate change \eqref{z-omega}, the metric $g_u$ takes the form
\begin{align*}
  g_u &= \frac{1}{\Omega_{u'}^2} 4 d\omega d \overline{\omega}, \\
  \Omega_{u^{\prime}}  & =  u^{\prime}_0 (1+ \omega \overline{\omega})
   +  u^{\prime}_1 (1 - \omega \overline{\omega} )
  + u^{\prime}_2 (\omega + \overline{\omega} ) + u^{\prime}_3 i (\omega - \overline{\omega}),
\end{align*}
where the constants   $u' := \{u_0^{\prime}, u_1^{\prime}, u_2^{\prime}, u_3^{\prime} \}$
are obtained from  $u = \{ u_0, u_1, u_2, u_3\}$ by the transformation
%%%%%%%%%%%%%%%%%%%%%%%%%%%%%%%%%%%%%%%%%%%%
%%% see file 2dimMetTrans_final.red %%%%%%%%
%%%%%%%%%%%%%%%%%%%%%%%%%%%%%%%%%%%%%%%%%%%%
\begin{align}
\epsilon  \left ( \begin{array}{c}
            u_0^{\prime} \\
            u_1^{\prime} \\
            u_2^{\prime} \\
            u_3^{\prime}
          \end{array}
  \right ) = & 
               \underbrace{ \frac{1}{2} \left (
    \begin{array}{cccc}
\al \alb + \be\beb + \ga\gab +\del\delb &
\al \alb - \be\beb + \ga\gab -\del\delb &
  - \al \beb - \be\alb - \ga\delb -\del\gab &
  i(  \al \beb - \be\alb + \ga\delb -\del\gab) \\
  \al \alb + \be\beb - \ga\gab -\del\delb & 
\al \alb - \be\beb - \ga\gab + \del\delb & 
 - \al \beb - \be \alb  + \ga\delb + \del \gab & 
 i(\al \beb - \be \alb  - \ga \delb + \del \gab) \\
 - (\al \gab + \be \delb + \ga \alb  +  + \del \beb) &
  - \al \gab + \be \delb - \ga \alb  + \del \beb & 
  \al \delb   + \be \gab + \ga\beb + \del \alb & 
  i( - \al \delb   + \be \gab - \ga \beb + \del \alb
  ) \\
  i ( - \al \gab - \be \delb + \ga \alb  
  + \del \beb) &
  i( - \al \gab + \be \delb + \ga \alb  
  - \del \beb) & 
  i(\al \delb 
  + \be \gab - \ga \beb - \del \alb ) &
\al \delb - \be \gab - \ga \beb   + \del \alb   
    \end{array}
\right )}_{\Lambda_{(\alpha, \beta,\gamma,\delta)}}  \left ( \begin{array}{c}
            u_0 \\
            u_1 \\
            u_2 \\
            u_3
          \end{array}
\right )  \label{LambdaMoeb}
\end{align}
where $\epsilon := \pm 1$. This sign reflects the impossibility (in general)
of choosing between $u$ and $-u$.  One can check that
$\Lambda_{(\alpha,\beta,\gamma,\delta)} = \M(\A^{-1})^{T}$ \eqref{MA} where $\A$ is as in \eqref{matrixA} and $^T$ denotes transpose. It follows that $\Lambda({\alpha,\beta,\gamma,\delta)}$  defines a morphism of groups  between $SL(2,\mathbb{C})$
  and $SO^{\uparrow}(1,3)$ and  that $u$ transforms as the components of a covector in the Minkowski
spacetime. Also observe that when $u$ is timelike or null (i.e. $\kappa_u \geq 0$), the
choice $u_0, u_0' >0$ selects $\epsilon=1$.

%We may particularize this to the
%Euclidean metric $g_E$ and to the standard spherical metric $g_{\mathbb{S}^2}$.% Under (\ref{z-omega}), they become, respectively
%\begin{align*}
%  g_E = 4 d\omega  \overline{\omega}  g_{\mathbb{S}^2} =
%\end{align*}

In order to express the metric in the coordinates $\{ v_2, v_2\}$ we need to compute the  functions $\omega \overline{\omega}$, $\omega + \overline{\omega}$
and $i(\omega- \overline{\omega})$ in terms of these variables.
For notational simplicity we introduce
the auxiliary quantities
\begin{align}
  h_1 := v_1 \cos \arg + v_2 \sin \arg, \quad \quad
  h_2 := v_2 \cos \arg - v_1 \sin \arg. \label{aux}
 \end{align}
 From \eqref{omega-zeta} with $\zeta_0=0$, a straightforward computation that uses basic  trigonometry yields
 %%%%%%%
 %%%%% See file
 %%%%%%%%%%%%%%%%%%%%%%%%%%%%%
 \begin{align*}
  \omega \overline{\omega}  & = 
  \frac{ Q \left (
  \cosh \left ( \sqrt{Q} h_2 \right ) + \cos \left ( \sqrt{Q} h_1 \right )
  \right ) }{4 
  \left (   \cosh \left ( \sqrt{Q} h_2 \right ) - \cos \left ( \sqrt{Q} h_1 \right ) \right )},  \\
  \omega + \overline{\omega} & =
   \frac{
  \sqrt{Q} \sin \arg \sinh \left ( \sqrt{Q} h_2 \right )
  - \sqrt{Q} \cos \arg \sin \left ( \sqrt{Q} h_1 \right ) }{\cosh \left ( \sqrt{Q} h_2 \right ) - \cos \left ( \sqrt{Q} h_1 \right )},    \\
  i \left ( \omega -  \overline{\omega} \right ) & =- 
\frac{
  \sqrt{Q} \cos \arg \sinh \left ( \sqrt{Q} h_2 \right )
                                                   + \sqrt{Q} \sin \arg \sin \left ( \sqrt{Q} h_1 \right ) }{\cosh \left ( \sqrt{Q} h_2 \right ) - \cos \left ( \sqrt{Q} h_1 \right )}.                                                  
\end{align*}
Since $d \omega = \frac{d \omega}{ d \zeta} d \zeta =  2 ( \omega^2 + \frac{Q}{4}
e^{-2i \arg} ) d \zeta$, determining the line-element
  $d \omega d \overline{\omega}$ requires expressing
  $| \omega^2 + Q/4 e^{-2 i \arg} |^2$ in terms of $\{v_1, v_2\}$. The result is
  obtained by a direct computation,
      \begin{align*}
    4 \left ( \omega^2 + \frac{Q}{4} e^{-2 i \arg} \right ) 
    \left ( \overline{\omega}^ 2 + \frac{Q}{4} e^{2 i \arg} \right )=
    \frac{Q^2}{\left ( \cosh ( \sqrt{Q} h_2 )  - \cos ( \sqrt{Q}
        h_1 ) \right )^2}.
      \end{align*}      Let us introduce the functions
\begin{align}
  f_{+} (v_1,v_2) &:= \frac{1}{4} \left (
                    \cosh ( \sqrt{Q} h_2 ) + \cos (  \sqrt{Q} h_1  ) \right ) \nonumber \\
  f_{-}(v_1,v_2) & := \frac{1}{Q} \left ( \cosh ( \sqrt{Q} h_2 )-
                   \cos ( \sqrt{Q} h_1  ) \right )
                   \nonumber  \\
   f_2(v_1,v_2)& := \frac{1}{\sqrt{Q}} \left ( \sin \arg \sinh ( \sqrt{Q} h_2  ) - \cos \arg \sin ( \sqrt{Q} h_1 ) \right ) \label{efes} \\
  f_3 (v_1,v_2) & := \frac{-1}{\sqrt{Q}} \left ( \cos \arg \sinh  ( \sqrt{Q} h_2 )   + \sin \arg \sin ( \sqrt{Q} h_1) \right ) \nonumber 
  \end{align}
so that we may express
\begin{align*}
  \omega \overline{\omega} =
  \frac{f_{+}}{f_{-}}, \quad \quad
  \omega + \overline{\omega} = \frac{f_2}{f_{-}},
  \quad \quad i \left ( \omega - \overline{\omega}  \right ) =
  \frac{f_3}{f_{-}}
  \end{align*}
  All these function admit smooth limits at $Q \rightarrow 0$, with corresponding expressions
  \begin{align*}
    f_+ (v_1,v_2) & =  \frac{1}{2} \\
    f_2 (v_1, v_2) & = -v_1 \\
                     %=  \sin \arg h_2 - \cos \arg h_1 
    f_3(v_2, v_2) & = - v_2\\
    %\cos \arg h_2 + \sin \arg h_1 
    f_{-}(v_1,v_2) & =   \frac{1}{2} \left ( v_1^2 + v_2^2 \right ).
%                     = \frac{1}{2} \left ( h_2^2 + h_1^2 \right )
                     \end{align*}
                     For $Q\neq0$, the functions $\{f_{+},f_{-},f_2, f_3\}$ are all periodic in the variable $h_1$ with periodicity $2\pi/\sqrt{Q}$. This corresponds to the fact that  the $\zeta$-plane is a covering of the $\omega$-plane, with the identification defined by the translation $\zeta_t$.

  Thus,  in the adapted coordinates $\{ v_1, v_2\}$ where $\xi = \partial_{v_1}$ and
  $\xi^{\perp} = \partial_{v_2}$, the metric $g_0 := 4 d \omega d \overline{\omega}$ takes the form
\begin{align*}
  g_0  = \frac{4}{f_{-}^2}  d \zeta d \overline{\zeta} = \frac{Q^2}{\left ( \cosh \left ( \sqrt{Q} h_2 \right ) - \cos \left ( \sqrt{Q} h_1 \right ) \right)^2}
  \left ( dv_1^2 + dv_2^2 \right ).
  \end{align*}
  Hence, the metric $g_u$ becomes
  \begin{align}
    g_u & = \frac{1}{\left ( (u_0' - u_1') f_{+} + (u_0' + u_1') f_{-}
          + u_2' f_2 + u_3' f_3 \right )^2} \left (d v_1^2 + dv_2^2 \right ) 
     :=
    \frac{1}{\widehat{\Omega}^2(v_1,v_2)} \left (d v_1^2 + dv_2^2 \right ). \label{guadapted}
  \end{align}
  We may now summarize the results {obtained so far concerning GCKV.
    \begin{theorem}
      \label{Main}
    Let $\mathbb{E}_2$ be the euclidean plane and $\{x,y\}$ be Cartesian coordinates. Let $\xi$ be a GCKV in this space and define the complex constants
    $\{\mu_0, \mu_1, \mu_2\}$ by means of the expression of $\xi$ given by
    \eqref{muform} in the complex
  coordinates $\z = \frac{1}{2} (x-iy)$, $\overline{\z}= \frac{1}{2}(x-iy)$. Define
  \begin{align*}
    \alpha = \frac{1}{2} \left ( \delta \mu_2 - \gamma \mu_1 \right ), \qquad
    \beta = \frac{1}{2} \delta \mu_1 - \gamma \mu_0,
  \end{align*}
  where $\gamma$ and $\delta$ are any pair of complex constants satisfying
  \begin{align*}
    \frac{1}{2} \delta^2 \mu_2 - \gamma\delta \mu_1 + \gamma^2 \mu_0 =1.
    \end{align*}
    Then
  $\xi$ takes its canonical form (c.f. Proposition \ref{transf})
  \begin{align*}
    \xi =
    \left ( \mu_0'
   % \frac{1}{4} \left (\aa_{\{\mu\}} - i \bb_{\{\mu\}} \right )
    + \omega^2 \right ) \partial_{\omega} +
    \left ( \overline{\mu_0'}
    %\frac{1}{4} \left (\aa_{\{\mu\}} + i \bb_{\{\mu\}} \right ) 
    + \overline{\omega}^2 \right ) \partial_{\overline{\omega}} , \qquad
    % \aa_{\{\mu\}} - i \bb_{\{\mu\}} :=
    4 \mu_0' := 2 \mu_0 \mu_2 - \mu_1^2 ,
    %\qquad  \aa_{\{\mu\}}, \bb_{\{\mu\}} \in \mathbb{R}
    \end{align*}
  in the coordinate system
  $\{ \omega,\overline{\omega}\}$
  defined by  $\omega = (\alpha \z + \beta)/(\gamma \z + \delta)$.
  Any other coordinate system
  $\{ \omega',\overline{\omega}'\}$
  where $\xi$ is in canonical form is related to $\{ \omega,\overline{\omega}\}$  by (c.f. Corollary \ref{inv1})
  \begin{align*}
    \omega' = \frac{\delta' \omega - \gamma' \mu_0'}{\gamma'  \omega + \delta'},
    \qquad \qquad  \delta'{}^2 + \mu_0' \gamma'{}^2 = 1.
  \end{align*}
  In addition, the real coordinates $\{ v_1, v_2\}$ defined by $\zeta := v_1 + i v_2$ together with \eqref{omega-zeta} and $4\mu_0' := \aa_{\{ \mu\}} - i \bb_{\{\mu\}} = Q e^{-2 i \theta}$  are adapted to
  $\xi$ and $\xi^{\perp} := \xi_{\{ i \mu\}}$ (c.f. Lemma \ref{orto}), namely 
  $\xi = \partial_{v_1}$ and $\xi^\perp = \partial_{v_2}$. 
  Moreover, the class of metrics \eqref{metgu} is written in adapted coordinates as \eqref{guadapted}. 
    \end{theorem}
  
  We mentioned above that the freedom in the coordinate change that brings
  $\xi$ into its canonical form can be translated into the freedom of a
  constant shift in the coordinates $\{ v_1, v_2\}$. Given $\{ \tilde{v}_1, \tilde{v}_2\}$ let $\tilde{h}_1$ and $\tilde{h}_2$ by defined exactly by the
  same expression as \eqref{aux} but with $\{ v_1,v_2\}$ replaced by
  $\{ \tilde{v}_1, \tilde{v}_2\}$. Similarly, we introduce
  four functions  $\{\tilde{f}_{+}(\tilde{v}_1, \tilde{v}_2)$, $
  \tilde{f}_{-}(\tilde{v}_1, \tilde{v}_2)$, $
  \tilde{f}_{2}(\tilde{v}_1, \tilde{v}_2)$, $
  \tilde{f}_{3}(\tilde{v}_1, \tilde{v}_2) \}$ by the same definition
  as \eqref{efes}, with $\{ h_1, h_2\}$ replaced  by $\{ \tilde{h}_1,
  \tilde{h}_2\}$. Let us now consider the coordinate change 
\begin{align}
\left \{
\begin{array}{l}
v_1 = \tilde{v}_1
  - \cos\arg \ell_1 + \sin\arg \ell_2  \\
  v_2 = \tilde{v}_2
 -  \sin\arg \ell_1 -  \cos\arg \ell_2 
\end{array} \right .
  \label{shiftconst}
  \end{align}
  where $\ell_1$ and $\ell_2$ are constants. Then $h_1 = \tilde{h}_1 - \ell_1$
  and $h_2 = \widetilde{h}_2 - \ell_2$ and we may relate the functions $\{ f\}$  written in terms of $\{\tilde{v}_1, \tilde{v}_2\}$ with the functions $\{\tilde{f}\}$. The  result is
%%%%%%%%%%%%
%%%% See file adapted4.red %%%
%%%%%%%%%%%%%%%%%%%%%%%%%%%%%%
  \begin{align}
    \left ( \begin{array}{c}
              2 f_{+} \\
              2 f_{-} \\
              f_2 \\
              f_3
            \end{array}
    \right )_{\tilde{v}_1, \tilde{v_2}}
    = &\left ( \begin{array}{cccc}
                \frac{1}{2} (\coh + \co ) & \frac{Q}{8} (\coh - \co) &
                                                                           -\frac{\sqrt{Q}}{2} \si & \frac{\sqrt{Q}}{2} \sih \\
                \frac{2}{Q} ( \coh - \co) & \frac{1}{2} ( \coh + \co) &
                                                                          \frac{2}{\sqrt{Q}} \si & \frac{2}{\sqrt{Q}} \sih \\
                \frac{1}{\sqrt{Q}}  ( \cos\arg \si - \sin\arg \sih ) &
                                                                      - \frac{\sqrt{Q}}{4} ( \cos\arg \si + \sin\arg \sih ) & \cos\arg \co & - \sin \arg \coh \\
                \frac{1}{\sqrt{Q}} ( \cos\arg \sih + \sin\arg \si) & \frac{\sqrt{Q}}{4} ( \cos\arg \sih - \sin\arg \si ) & \sin\arg \co & \cos\arg \coh                                                             
      \end{array} \right ) \nonumber \\
&    \left ( \begin{array}{cccc}
               1 & 0 & 0 & 0 \\
               0 & 1 & 0 & 0 \\
               0 & 0 & \cos \arg & \sin \arg \\
               0 & 0 & -\sin \arg & \cos \arg
             \end{array}
                                      \right )
                                      \left ( \begin{array}{c}
                                                2 \tilde{f}_+ \\
                                                2 \tilde{f}_- \\
                                                \tilde{f}_1 \\
                                                \tilde{f}_2
                                              \end{array}
    \right ) := W(\ell_1,\ell_2)
    \left ( \begin{array}{c}
                                                2 \tilde{f}_+ \\
                                                2 \tilde{f}_- \\
                                                \tilde{f}_1 \\
                                                \tilde{f}_2
\end{array} \right ), \label{matricform}
                                              \end{align}
                                              where for notational simplicity we have introduced $\co = \coexp, \coh = \cohexp, \si = \siexp, \sih = \sihexp$.
                                              If we compare $W(\ell_1, \ell_2)$
                                              and $\T(\lambda_2, \lambda_3, \epsilon)$ we see that the matrices are identical after setting
                                              \begin{align}
                                                \lambda_2 = \frac{1}{\sqrt{Q}}
                                                \siexp, \quad \quad
                                                \lambda_3 =  \frac{1}{\sqrt{Q}}
                                                \sihexp, \quad \quad
                                                \epsilon \sqrt{1 - Q \lambda_2^2 } = \coexp. \label{lam-ell}
                                              \end{align}
                                              Of course this does not happen by chance. We have seen before that the shift in $\zeta$ corresponds to the subgroup
                                              of M\"obius transformation that leaves the canonical form of $\xi$ invariant. By the relationship between GCKV
                                              and skew-symmetric endomorphism
                                              in $\mathbb{M}^{1,3}$ described
                                              in Section \ref{GCKV_and_F}, this M\"obius subgroup corresponds to the set of orthochronous Lorentz  transformations that leave the skew-symmetric endomorphism invariant, and this is precisely
                                              the group $\{ \T(\lambda_2,
                                              \lambda_3, \epsilon) \}$.  With the choice we have made of the shift constants \eqref{shiftconst}, the relationship between the parameters $\{ \ell_1, \ell_2 \}$ and
                                              $\{ \lambda_2, \lambda_3 \}$ take the remarkably simple form given by \eqref{lam-ell}. Note that the map
                                              $(\ell_1, \ell_2) \rightarrow 
                                              ( \lambda_2, \lambda_3, \epsilon)$
                                              is again a covering. If we let $\ell_2$ be periodic with periodicity $\frac{2\pi}{\sqrt{Q}}$, the map is a bijection. Observe that, to make the comparison work, we have inserted a factor $2$
                                              in front of $f_{\pm}$ in the
                                              column vector \eqref{matricform}. The reason is easy to understand. The constants $\{ u_0', u_1', u_2', u_3'\}$ in the
                                              conformal factor $\widehat{\Omega}$ in the metric $g_u$ define a Lorentz covector  of length $-u_0'{}^2 +
                                              u_1'{}^2 + u_2'{}^2 + u_3'{}^2 
= - (u_0' + u_1') (u_0' - u_1') +  u_2'{}^2 + u_3'{}^2$. This means that, viewed as vectors in a Lorentz space, the basis $\{ f_{+}, f_-, f_{2}, f_{3}\}$ is semi-null, but with scalar product $\la f_+, f_{-} \ra = \frac{1}{2}$ However,  the transformation law $\T(\lambda_2, \lambda_3, \epsilon)$ was written in a semi-null basis
$\{ \ell, k , e_2, e_3\}$ with normalization $\la \ell, k \ra =-2$, which is precisely the normalization of the basis $\{ 2 f_+, 2 f_{-}, f_2, f_3\}$.

                                              Having obtained the transformation law for $\{ f_{+}, f_{-}, f_2, f_3\}$ it follows
                                              immediately that under the coordinate transformation \eqref{shiftconst}, the metric $g_u$ becomes
                                                \begin{align*}
    g_u =  \frac{1}{\left ( (\tilde{u}_0 - \tilde{u}_1) \tilde{f}_{+} + (\tilde{u}_0 + \tilde{u}_1 ) \tilde{f}_{-}
    + \tilde{u}_2 \tilde{f}_1 + \tilde{u}_3 \tilde{f}_2 \right )^2} \left (d \tilde{v}_1^2 + d\tilde{v}_2^2 \right ) 
    \end{align*}
    where the constants $\{ \tilde{u}_0, \tilde{u}_1,
    \tilde{u}_2, \tilde{u}_3 \}$ are given by
    \begin{align*}
      \left ( \begin{array}{c}
                \frac{1}{2} \left (\tilde{u}_0 - \tilde{u}_1 \right ) \\
                \frac{1}{2} \left ( \tilde{u}_0 + \tilde{u}_1 \right ) \\
                \tilde{u}_2 \\
                \tilde{u}_3                
                \end{array}
                \right ) = \epsilon (W (\ell_1, \ell_2))^T 
      \left ( \begin{array}{c}
                \frac{1}{2} \left ( u'_0 - u'_1 \right ) \\
\frac{1}{2} \left ( u'_0 + u'_1 \right ) \\
                u'_2 \\
                u'_3 
                \end{array}
               \right ) 
    \end{align*}
    (the reason for the sign $\epsilon$ is the same as discussed before).

  \section{Applications}\label{secapps}

  \subsection{Killing vectors of $g_u$}

  Our aim is to determine under which conditions $\xi$ is a Killing vector
  of the metric $g_u$. We will address the question by analyzing the situation in the adapted coordinates. Since $\xi= \partial_{v_1}$,
  $\xi$ will be a Killing vector of $g_u$ if and only if the function
  $\widehat{\Omega}$ satisfies $\partial_{v_1} \widehat{\Omega} =0$. It is
  straightforward to check that
  %%%%%%%%%%%%%%%%%%%%%%%%%%%%%%%%%%%%%%%%%%%
  %%% See file adapted2.red
  %%%%%%%%%%%%%%%%%%%%%%%%%%%%%%%%%%%%%%%%%%%
    \begin{align*}
 \partial_{v_1} f_{+} & =  \frac{Q}{4} \left (  \cos(2 \arg) f_2 +
                           \sin(2 \arg) f_3 \right )                          , \\
      \partial_{v_1} f_{-} & = - f_2, \\
     \partial_{v_1} f_{2} & = - 2 f_{+} + \frac{Q}{2} \cos (2 \arg) f_{-}, \\
    \partial_{v_1} f_{3} & =  \frac{Q}{2} \sin( 2 \arg) f_{-}, 
  \end{align*}
 which imply
  \begin{align*}
    \partial_{v_1} \widehat{\Omega}   = & 
        - 2 u_2'  f_+
    +  \frac{Q}{2}  \left ( \cos(2\arg) u_2' + \sin(2 \arg) u_3'
    \right ) f_{-}                                         + \left ( \frac{Q}{2} \cos (2\arg) u_-
                                       - 2 u_{+}'\right ) f_2
                        + \frac{Q}{2} \sin (2 \arg) u_-' f_3,
  \end{align*}
  where we have set $u_{\pm}' := \frac{1}{2} ( u'_{0} \pm u'_{1})$.
    The functions $\{ f_{+}, f_{-}, f_2, f_3\}$ are linearly independent, so this
  derivative will vanish if and only if each coefficient vanishes. If
  $Q \sin (2 \arg) \neq 0$, it is immediate that the only solution
  is $u_+'=u_{-}'=u_2'=u_3'=0$, which is not possible for a metric $g_u$. Thus, a necessary condition for $\xi$ to be a Killing vector of (any) $g_u$
  is that the invariant (see \eqref{defQsigma}) $\aa_{\{\mu\}} - i \bb_{\{\mu\}}$ be
  real (i.e. $\bb_{\{\mu\}}=0$). When $Q \neq 0$, the condition $\sin (2 \arg)=0$ is
  $\arg  \in \{ 0, \frac{\pi}{2}\}$ (recall that $\arg \in [0,\pi)$ by construction). To cover all cases at once we set
  $\cos\arg = \hep$ and $\sin \arg = 1- \hep$, with $\hep^2 = \hep$. Then
  $\cos (2 \arg ) = 2 \hep - 1$ (this choice is also valid when $Q=0$ because
  $\arg$ can be fixed to any value).
  Then
    \begin{align*}
    \partial_{v_1} \widehat{\Omega} =0
    \quad \Longleftrightarrow
    \quad     (u_{-}',u_{+}', u_2',u_3') = s_1 \underbrace{\left ( 1 , \frac{Q}{4} (2 \hep -1) , 0, 0 \right )}_{w_1}
    + s_2 \underbrace{(0,0,0,1)}_{w_2}, \quad                                    s_1,s_2 \in \mathbb{R}. 
  \end{align*}
  The Lorentzian norm of this vector is $-4 u_+' u_-' + u_2'{}^2 +u_3'{}^2 =
  -  (2 \hep -1) Q s_1^2 + s_2^2$. Under the constant shift given by $\ell_1, \ell_2$, the two-dimensional vector space spanned by $w_1$ and $w_2$ remains invariant, and the vector $s_1 w_1 + s_2 w_2$ transforms to $\tilde{s}_1 w_2 +
  \tilde{s}_2 w_3$ with
  %%%%%%%%%%%%%%%%%%%%%%%%%%%%%
  %%% See adapted4.red
  %%%%%%%%%%%%%%%%%%%%%%%%
  \begin{align*}
    \left ( \begin{array}{c}
      \tilde{s}_1 \\
      \tilde{s}_2
    \end{array}
    \right ) =\epsilon  \left ( \begin{array}{cc}
      \hep \cosh (\sqrt{Q} \ell_2 ) + 
      \cos (\sqrt{Q} \ell_1 )(1 - \hep )  &
      \frac{1}{\sqrt{Q}} \left ( \sinh(\sqrt{Q} \ell_2) \hep + \sin ( \sqrt{Q} \ell_1) ( 1- \hep) \right ) 
            \\
\sqrt{Q} \left ( \sinh(\sqrt{Q} \ell_2) \hep - \sin ( \sqrt{Q} \ell_1) ( 1- \hep) \right ) &   \hep   \cosh (\sqrt{Q} \ell_2 ) + 
\cos (\sqrt{Q} \ell_1) (1 - \hep ) 
          \end{array}
    \right )
    \left ( \begin{array}{c}
      s_1 \\
      s_2
    \end{array}
    \right ).
  \end{align*}
  This transformation leaves the norm  $-( 2\hep -1) Q s_1^2 + s_2^2$ invariant (as it must) and defines a group which is one-dimensional when $Q \neq 0$
  and two-dimensional when $Q=0$. Thus, when  transforming the vector
  $u$ into the original coordinate system $\{ z, \overline{z}\}$ we may ignore the action of the invariance group that
  leaves the canonical form of $\xi$ invariant provided we let $u$ take all non-zero values in the vector  space $\spn{w_1,w_2}$ 
%   $\{ s_1 w_1 + s_2 w_2\}$
  .  We may summarize the result in the following theorem.
%%%%%%%%%%%%%%%%%%%%%%%%%%%%%%%%%%
  %%% see file theorem2.red
  %%%%%%%%%%%%%%%%%%%%%%%%%%%%%
  \begin{theorem}
    \label{Killing}
    Given a non-identically zero GCKV $\xi$ in two-dimensional Euclidean space  and let  $\{\mu\}:= \{ \mu_0,\mu_1,\mu_2\}$ be the set of parameters
    such that $\xi = \xi_{\{\mu\}}$ in the coordinate system $\{ z, \overline{z}\}$. Let  $U \subset \mathbb{R}^4 \setminus \{ 0\} $ be defined by the property that for all $u \in U$,  $\xi$ is a Killing vector of the
    metric $g_u$ (defined in \eqref{metgu}). Then
    \begin{itemize}
    \item If  $2 \mu_0 \mu_2 - \mu_1^2 \not \in \mathbb{R}$  then  $U = \emptyset$.
      \item If $2 \mu_0 \mu_2 - \mu_1^2  \in \mathbb{R}$, let $\delta, \gamma$ be any pair of complex numbers satisfying 
      \begin{align*}
          \frac{1}{2} \delta^2 \mu_2 - \gamma \delta \mu_1 + \gamma^2 \mu_0 =1
      \end{align*}
      and set $\alpha = \frac{1}{2} ( \delta \mu_2 - \gamma \mu_1)$
      and
      $\beta = \frac{1}{2} \delta \mu_1 - \gamma \mu_0$. Then $u \in U$ if and only if
      \begin{align*}
        \left ( \begin{array}{c}
          u_0 \\
          u_1 \\
          u_2 \\
          u_3
        \end{array}
        \right )
        = \M(\A)^T 
        \left ( \begin{array}{c}
          s_1 \left ( \frac{1}{4} ( 2 \mu_0 \mu_2 - \mu_1^2 ) + 1\right ) \\
          s_1 \left (\frac{1}{4} ( 2 \mu_0 \mu_2 - \mu_1^2 ) - 1\right ) \\
          0 \\
          s_2 \\
                  \end{array}
        \right )
    \end{align*}
    where $(s_1, s_2) \in \mathbb{R}^2 \setminus \{0\}$, $\A$ is the matrix
    \eqref{matrixA}  and $\M(\A)$ was defined in \eqref{MA}. 
    \end{itemize}
    Moreover, such $g_u$ has constant curvature $\kappa_u$ given by
    \begin{align*}
      \kappa_u =  s_1^2 ( 2 \mu_0 \mu_2 - \mu_1^2 )
      - s_2^2.
\end{align*}      
  \end{theorem}
      \begin{proof} We only need to check that
      $w_1 = ( 1, \frac{1}{4} (2 \mu_0 \mu_2 - \mu_1^2), 0, 0)$, This is an
      immediate consequence of the definitions
      \eqref{defQsigma} and \eqref{aabb}, which
      in the case $\cos \arg = \hep$ and $\sin \arg  = 1 - \hep$
      imply
            \begin{align*}
        Q\left (2 \hep - 1 \right ) =
        2 \mu_0 \mu_2 - \mu_1^2.
      \end{align*}
      \end{proof}
      
      One may wonder why this problem has no been addressed in the original coordinate
      system $\{ z, \overline{z} \}$. The Lie derivative of a metric
      $g_{\Psi} := 4 \Psi^{-2} dz d\overline{z}$ along $\xi_{\{\mu\}}$ (given by
      \eqref{muform}) is
      \begin{align*}
        \pounds_{\xi_{\{\mu\}}} g_{\Psi} = 
        \left (
        - 2 \xi_{\{\mu\}} (
          \Psi) + \Psi \left ( \mu_1 + \overline{\mu_1} + \mu_2 z + \overline{\mu_2} \overline{z} \right ) \right ) g_{\Psi}.
        \end{align*}
        Thus $\xi_{\{\mu\}}$ is a Killing vector of $g_u$ if and only if
        \begin{align*}
- 2 \xi_{\{\mu\}} ( \Omega_u ) + \Omega_u \left ( \mu_1 + \overline{\mu_1} + \mu_2 z + \overline{\mu_2} \overline{z} \right ) =0.
        \end{align*}
        The computation gives a polynomial  in $\{ z, \overline{z}\}$ of degree two. Equating each coefficient to zero, one finds that the conditions that need to
        be satisfied can be written in the form
%%%%%%%%%%%%%%%%%%%%%%%%%%%%%%%%%%%%%%%
        %%% see file Killing.red %%%%
        %%%%%%%%%%%%%%%%%%%%%%%%%%%%%%%
        \begin{align}
\left ( 
          \begin{array}{cccc}
            0 & - \nu & -  a_x + \frac{b_x}{2} & -a_y + \frac{b_y}{2} \\
            -\nu & 0 & - a_x - \frac{b_x}{2} & -a_y - \frac{b_y}{2}  \\
            - a_x + \frac{b_x}{2} & a_x + \frac{b_x}{2} &  0   & - \omega \\
            -a_y + \frac{b_y}{2} & a_y + \frac{b_y}{2} &  \omega & 0           \end{array}
          \right )
          \left ( \begin{array}{c}
                    -u_0 \\
                    u_1 \\
                    u_2 \\
                    u_3 
                  \end{array}
          \right ) =
          \left ( \begin{array}{c}
                    0 \\
                    0 \\
                    0 \\
                    0 
                  \end{array}
          \right )         \label{condsKill}      
        \end{align}
        where we have expressed $\{ \mu \}$ in terms of its  real and imaginary parts by means of \eqref{mudata}. Recalling the relationship between
        GCKV $\xi$  and skew-symmetric endomorphisms $F_{\xi}$ we conclude that
        $\xi_{\{\mu\}}$ is a Killing vector of $g_u$ if and only if the non-zero Lorentz vector
        $(-u_0, u_1, u_2, u_3)$ lies in the kernel of $F_{\xi}$ (observe that this vector is obtained from the covector $u$ by raising indices with the
        Minkowski metric). Being skew-symmetric  and not identically zero, $F_{\xi}$ can only have rank two or four, so in order to admit a non-trivial kernel, the rank must be two. This corresponds to the condition $\bb_{\{\mu\}} =0
        \Longleftrightarrow  \mbox{Im} (2 \mu_0 \mu_2 - \mu_1^2)=0$. So, the kernel is two-dimensional, which recovers the statement in Theorem \ref{Killing} that the
        set $U \cup \{ 0 \}$ is a two-dimensional vector space. Thus, the problem becomes geometrically very neat in the original coordinate system.  However, in
        Theorem \ref{Killing} we have been able to determine explicitly the vector subspace $U \cup \{ 0 \}$ (equivalently the kernel of $F_{\xi}$, after index raising) in a way that covers all cases at once. It is not so clear how to achieve the same
        by a direct attempt of solving \eqref{condsKill}  in such a way that
        the solution covers all possible values of  $\{b_x, b_y, \nu, \omega, a_x, a_y\}$ under the restriction
        $b_x a_y - b_y a_x + \nu \omega =0$ (namely
        $\mbox{Im} (2 \mu_0 \mu_2 - \mu_1^2 )=0$).

        The issue addressed in Theorem \ref{Killing} is to determine for which metrics $g_u$ a given GCKV is Killing. A complementary problem is to fix $g_u$ and determine all GCKV which are Killings of $g_u$. This problem may be approached
        in the language of skew-symmetric endomorphisms. A skew-symmetric endomorphism $F$ in $\mathbb{M}^{1,3}$  of rank two is necessarily  of the form
        $F = q_1 \otimes \bm{q_2} - q_2 \otimes \bm{q_1}$ where
        $q_1$ and $q_2$ are linearly independent Lorentz vectors
        and boldface denote the metrically related one-form. A vector
        $u$ lies in the kernel of $F$ if and only if it is orthogonal to $q_1$
        and $q_2$. Thus, the set of Killing vectors of $g_u$ is obtained from
        all skew-symmetric endomorphisms
        \begin{align*}
          F_{u^{\perp}} := \{ F = q_1 \otimes \bm{q}_2 - q_2 \otimes \bm{q_1}; \quad \quad 
          \mbox{span} \{ q_1, q_2\} = u^{\perp} \}.
        \end{align*}
        where $u^{\perp}$ stands for the set of vectors in the kernel of the covector
        $(u_0, u_1, u_2, u_3)$.        We do not attempt to find an explicitly parametrization of all Killing
        vectors of $g_u$ that covers at once all possible choices of $u$ (this problem does not appear to be simple either in terms of endomorphisms, or by using canonical forms of $\xi$).

%        One checks that the rank of the metric is 3 for any value of $u = \{ u_0, u_1, u_2, u_3\}$. In fact, if we name the row vectors as $\{ l_1, l_2, l_3, l_4\}$, they satisfy the linear relation
%        $( u_0 - u_3) l_1 + (u_0 + u_3) l_2 - u_1 l_3 - u_2 l_4 =0$. Conditions \eqref{condsKill} can be equivaliently written as
%        \begin{align*}
%\left ( 
%          \begin{array}{cccc}
%             & &  &   \\
%             & &  &  \\
%             & &  &   \\
%             & &  &  
%          \end{array}
%          \right )
%          \left ( \begin{array}{c}
%                    u_0 \\
%                    u_1 \\
%                    u_2 \\
%                    u_3 
%%                  \end{array}
%          \right ) =
%          \left ( \begin{array}{c}
%                    0 \\
%                    0 \\
%                    0 \\
%                    0 
%                  \end{array}
%          \right )               
%        \end{align*}
% %         
% % 
% %           
% %       
% % 
% %       
  \subsection{Transverse and traceless and Lie constant
    tensors on $\mathbb{E}^2$}

  Transverse and traceless (TT) symmetric 2-covariant tensors, namely, tensors
  $D_{\alpha\beta} = D_{\beta\alpha}$ satisfying
  (indices are raised with a metric $g$ and $\nabla$ is the
  corresponding Levi-Civita connection)
  \begin{align*}
    \nabla_{\alpha}  D^{\alpha\beta}=0  \quad \mbox{ (transverse)}, \quad \quad D^{\alpha}{}_{\alpha} =0 \quad \mbox{ (traceless)} 
  \end{align*}
    play a prominent role in General Relativity, in several circumstances. For example, they are fundamental for the construction of initial data in spacelike slices with prescribed regularity at spacelike infinity \cite{DaiFri2001} or black hole initial data \cite{Beig1997}.
    Another example is the free data at null infinity for $\Lambda$-vacuum spacetimes with positive cosmological constant (see the original work \cite{Fried86initvalue} or more modern reviews \cite{Friedrich2002}, \cite{Friedrich2014}). In this setup, an interesting subclass that arises
  when the spacetime admits Killing vectors is the subclass of TT tensors
  which satisfy the so-called Killing initial data (KID) equation \cite{KIDPaetz}. In dimension $n$, this equation
  is
  \begin{align*}
    \pounds_{\xi} D_{\alpha\beta} + \frac{n-2}{n} (\mbox{div}_g \xi) D_{\alpha\beta} = 0
  \end{align*}
  where $\xi$ is a conformal Killing vector of $g$ and $\pounds_{\xi}$,
  $\mbox{div}_g \xi$ stand respectively for the Lie derivative along $\xi$ and
  the divergence of $\xi$ with respect to $g$. In dimension $n=2$ the general solution of (local) TT tensors satisfying
  the KID equation can be explicitly solved. Although this dimension
  is not particularly interesting from a physical point of view, there are several motivations for presenting the result. Firstly, dimensional reduction is
  a useful tool in many geometric problems, so it is not unlikely that
  the case of dimension two may find applications in higher dimensions. Also, the $n=2$ case may serve as a toy model to address the (much more difficult) problem in higher dimensions. In addition.  the solution we find turns out to admit an interesting generalization in arbitrary dimension. And lastly, it is remarkable, that the problem is so simple in dimension $n=2$ that its general solution can be explicitly given.   
      
  A key property of the TT conditions and of the KID equations is their
  conformal covariance. If $D_{\alpha\beta}$ is a TT tensor with respect to $g$
  then $\Psi^{2-n} D_{\alpha\beta}$ is a TT tensor with respect to
  $\Psi^2 g$. Also, if $D$ satisfies the KID equation for $g$, then $\Psi^{2-n} D$  also satisfies the KID equation for $\Psi^2 g$. In dimension $n=2$ one actually has conformal invariance. Since all two-dimensional metrics are locally
  conformal to the flat metric, and we are interested in solving the (more general)  local problem, we may assume that $g = 4 dz d \overline{z}$. As already mentioned, a vector field $\xi$ is conformal of this metric if and only if
  $\xi = f(z) \partial_z + \overline{f}(\overline{z}) \partial_{\overline{z}}$.
  We expand $D = D_{zz} dz^2 + D_{\zb\zb} d\zb^2 + 2 D_{z\zb} dz d\zb$.  The condition of being traceless is $D_{z \zb} =0$ and $D$ real requires $D_{\zb\zb} = D_{zz}$,  With these restrictions, the transverse
  equations take the following explicit and simple form
    \begin{align*}
      \partial_z D_{\zb \zb}=0, \quad \quad \partial_{\zb} D_{zz} = 0,
    \end{align*}
    so $D_{zz}$ is a holomorphic function of $z$. Imposing 
    transverse and traceless as well as the reality condition, the KID equations read
    \begin{align*}
      f \frac{D_{zz}}{dz} + 2 D_{zz} \frac{df}{dz} =0,
    \end{align*}
    which integrates to $D_{zz} = \frac{q}{ f^2}, q \in \mathbb{C}$. Writing
    $q = q_1 + i q_2$, with real $q_1, q_2$, we conclude that the most general
    (real) TT tensor that satisfies the KID equation is a linear combination of (we add the factor $4$ for convenience)
    \begin{align*}
      D_1 := \frac{1}{4} \left ( \frac{1}{f^2} dz^2 + \frac{1}{\overline{f}^2} d\zb^2 \right ),
      \quad \quad
      D_2 = \frac{i}{4} \left ( \frac{1}{\overline{f}^2} d\zb^2 -
      \frac{1}{f^2} dz^2 \right ).
    \end{align*}
    These expressions are valid in the coordinate system $\{ z, \overline{z}\}$. We are interested in covariant expressions that are valid in any coordinate system, and are explicitly invariant under conformal transformations. To achieve this,
    we introduce the vector field
    \begin{align}
      \xi^{\perp} := i \left ( f \partial_z - \overline{f} \partial_{\zb} \right ). \label{xiperp}
    \end{align}
        This is everywhere orthogonal to $\xi$ and
    has the same norm at every point. If the zeros of $\xi$ do not separate
    the manifold, these two properties define
    $\xi^{\perp}$ in terms of $\xi$ uniquely except for a global sign. If the zeroes
    of $\xi$ separate the manifold, $\xi^{\perp}$ is still uniquely defined (up
    to a sign) if one adds the condition that $\xi^{\perp}$ is a
    conformal Killing vector of $g$ (which \eqref{xiperp} clearly is). Thus, we may speak of $\xi^{\perp}$ unambiguously (up to global sign), once $\xi$ has been fixed. Next we note that, in the $\{z, \zb\}$ coordinate system and with
    respect to the metric $g_E := 4 dz d\zb$ we have
    \begin{align*}
      \bm{\xi} = 2 f d\zb + 2 \overline{f} d\zb, \quad \quad
      |\xi|^2_{g_E} := g_E (\xi, \xi )= 4 f \overline{f}, \qquad \qquad 
      \bm{\xi^{\perp}} = 2i f d \zb - 2 i \overline{f} dz,
      \quad \quad |\xi^{\perp}|^2_{g_E} = 4 f \overline{f}
    \end{align*}
        and then we may write
    \begin{align*}
      D_1 & = \frac{1}{|\xi|^4_g}
      \left ( \bm{\xi} \otimes \bm{\xi} -
      \frac{1}{2} |\xi|^2_{g_E} g_E \right ), \\
      D_2 & = \frac{1}{2 |\xi|^4_{g_E}} \left ( \bm{\xi} \otimes \bm{\xi^{\perp}}
      + \bm{\xi^{\perp}} \otimes \bm{\xi} \right ).
    \end{align*}
    These expressions are obviously coordinate independent and also conformally
    invariant. Thus, $D_1$ and $D_2$ take this form also for the original metric
    $g$. Summarizing, we have proved the following theorem.
    \begin{theorem}\label{theoTT}
      Let $(M,g)$ be a two-dimensional Riemannian manifold and $\xi$ a conformal Killing vector of $g$. Let $D$ be a (real) transverse and traceless symmetric, $2$-covariant tensor that satisfies the KID equation with respect to $\xi$.
      Then $D$ is a linear combination (with constants) of
      \begin{align*}
        D_{\xi} & := \frac{1}{|\xi|^4_g} \left (  \bm{\xi} \otimes \bm{\xi} -
        \frac{1}{2} |\xi|^2_{g} g \right ), \\
        D_{\xi, \xi^{\perp}} & := \frac{1}{2 |\xi|^2_g |\xi^{\perp}|^2_g} \left (
        \bm{\xi} \otimes \bm{\xi^{\perp}} +
        \bm{\xi^{\perp}} \otimes \bm{\xi} \right ),
      \end{align*}
      where $\xi^{\perp}$ is defined as described above and
      $\bm{\xi} := g (\xi, \cdot)$,
      $\bm{\xi^{\perp}} := g (\xi^{\perp}, \cdot)$.
    \end{theorem}

\section*{Acknowledgements}

The authors acknowledge financial support under the projects
PGC2018-096038-B-I00
(Spanish Ministerio de Ciencia, Innovaci\'on y Universidades and FEDER)
and SA083P17 (JCyL). C. Pe\'on-Nieto also acknowledges the Ph.D. grant BES-2016-078094 (Spanish Ministerio de Ciencia, Innovaci\'on y Universidades).

           \printbibliography

\end{document}